\documentclass[reqno,10pt]{amsart}
\usepackage{amssymb,amsmath,latexsym,amsthm}
\usepackage{color,mathrsfs}
\usepackage{url}
\usepackage{enumerate}
\usepackage[square,sort,comma,numbers]{natbib}
\usepackage{dsfont}
\usepackage{diagbox}
\usepackage{graphicx}
\usepackage{colortbl}
\usepackage{float}
\usepackage{color}
\usepackage[colorlinks=true, linkcolor=blue_links,
urlcolor=blue_links, citecolor=blue_links]{hyperref}
\definecolor{blue_links}{RGB}{13,0,180} 
\definecolor{lightblue}{rgb}{0.9,0.9,1}

\usepackage{endnotes}

\renewcommand{\textbf}[1]{\begingroup\bfseries\mathversion{bold}#1\endgroup}

\usepackage{fancyhdr}
\usepackage{pstricks,pst-plot,pst-node,pstricks-add}

\newtheorem{thm}{Theorem}[section]

\newtheorem{corollary}[thm]{Corollary}

\newtheorem{proposition}[thm]{Proposition}

\theoremstyle{definition}
\newtheorem{remark}[thm]{Remark}

\newcommand{\argmin}{\mathop{\rm argmin}\nolimits}

\definecolor{orange-red}{rgb}{1.0, 0.27, 0.0}

\newcommand{\R}{\mathbb R}

\newcommand{\Z}{\mathbb Z}

\numberwithin{equation}{section}

\def\XXint#1#2#3{{\setbox0=\hbox{$#1{#2#3}{\int}$}
    \vcenter{\hbox{$#2#3$}}\kern-.5\wd0}}

\newcommand{\EEE}{\color{black}}

\newcommand{\ff}{\rho}

\allowdisplaybreaks

\begin{document}
\title{Lattice ground states for embedded-atom models in 2D and 3D}

\author{Laurent B\'{e}termin}
\address[Laurent B\'{e}termin]{Faculty of Mathematics, University of
  Vienna, 
Oskar-Morgenstern-Platz 1, 1090 Vienna, Austria.}
\email{laurent.betermin@univie.ac.at}
\urladdr{https://sites.google.com/site/homepagelaurentbetermin/}
 
 \author{Manuel Friedrich}
\address[Manuel Friedrich]{Applied Mathematics,  
University of M\"{u}nster, 
 Einsteinstr. 62, D-48149 M\"{u}nster, Germany.}
\email{manuel.friedrich@uni-muenster.de}
\urladdr{https://www.uni-muenster.de/AMM/en/Friedrich/}

 \author{Ulisse Stefanelli}
\address[Ulisse Stefanelli]{Faculty of Mathematics, University of Vienna,
Oskar-Morgenstern-Platz 1, 1090 Vienna, Austria,
 Vienna Research Platform on Accelerating Photoreaction Discovery,
 University of Vienna, W\"ahringerstra\ss e 17, 1090 Wien, Austria, and
 Istituto di Matematica Applicata e Tecnologie Informatiche {\it E. Magenes} - CNR
 via Ferrata 1, 27100 Pavia, Italy. }
\email{ulisse.stefanelli@univie.ac.at}
\urladdr{http://www.mat.univie.ac.at/$\sim$stefanelli}

\begin{abstract}
The Embedded-Atom Model (EAM) provides a phenomenological
description of
atomic arrangements in metallic systems. It consists of a
configurational energy depending on atomic positions and featuring the interplay of two-body atomic interactions
and nonlocal effects due to the corresponding electronic clouds. The
purpose of this paper is to mathematically investigate the
minimization of 
the EAM energy among   lattices  in two and three dimensions. We
present a suite of analytical and numerical results under 
different reference choices for the underlying interaction
potentials. In particular, Gaussian, inverse-power, and
Lennard-Jones-type interactions are addressed. 
\end{abstract}

\subjclass[2010]{70G75, 
  74G65, 74N05}

\keywords{Embedded-atom model, lattice energy minimization, Epstein zeta function.}

\maketitle

\section{Introduction}

Understanding the structure of matter is a central scientific and
technological quest, cutting
across disciplines and motivating an ever increasing computational effort.
First-principles calculations deliver accurate predictions but are
often impeded by the inherent quantum complexity, as systems size up \cite{CondensMatter}. One is hence
 led  to consider a range of approximations. The minimization of
empirical atomic pair-potentials represents the simplest of such
approximations being able to describe specific 
properties of large-scaled atomic systems. Still, atomic
pair-interactions fall short of describing the basic nature of
metallic bonding, which is multi-body by nature, and often deliver 
inaccurate  predictions  of metallic systems.

The {\it Embedded-Atom Model} (EAM) is a semi-empirical, many-atom
potential aiming at describing the atomic structure of metallic
systems by including a nonlocal electronic correction. Introduced by
Daw and Baskes \cite{DawBaskes83}, it has been used to  address  efficiently different
aspects inherent to atomic arrangements   including defects, dislocations, fracture, grain boundary structure
and energy, surface structure, and epitaxial growth. Proving
capable of reproducing experimental observations and being relatively
simple to implement, the Embedded-Atom Model  is now routinely used in molecular dynamic
simulations \cite{DawFoilesBaskeReviewEAM,LeSar}.  In particular,
it has been applied in a variety of metallic systems \cite{Foiles}, including alkali
metals Li, Na, K \cite{DorrellPartay,JohnsonOh,YuanChenShen},
transition metals Fe, Ni, Cu, Pd, Ag, Pt, Au \cite{CaiYe,GrocholaGold,JohnsonOh,LeSar},
post-transition metals Al, Pb \cite{CaiYe,JaffeAl,SuttonChen}, the metalloid Si
\cite{BaskesSi}, and some of their alloys \cite{CaiYe,Johnson}.\EEE 

 In the case of a metallic
system with a single atomic species, the
{\it EAM energy} is specified as
$$\sum_{i} F(\overline \ff_i) +
\sum_{i\not =j} \phi(|x_i - x_j|)\quad \text{with}\quad \overline
\ff_i =\sum_{j\not =i}  \ff  (|x_i - x_j|).$$
Here, $\{x_i\}$ indicate atomic positions in ${\mathbb R}^d$ and the
 long-range 
{\it interaction potential} $\phi\colon\R_+:=(0,\infty)\to \R_+$ modulates atomic
pair-interactions. Atomic positions induce electronic-cloud
distributions. The function $ \ff \colon \R_+ \to \R_+$ models the long-range {\it electron-cloud contribution} of an atom placed at $x_j$ on an atom placed at
$x_i$.  The sum  $\overline \ff_i$
describes the  cumulative  effect on the atom placed at
$x_i$
of the electronic clouds  related to all other atoms.  Eventually,
the function $F\colon\R_+ \to\R_+$ describes the energy needed
to place (embed) an atom at position $x_i$ in the host electron gas
created  by the other atoms at positions  $\lbrace x_j \rbrace$.  

Purely pair-interaction potentials can be re-obtained from the EAM
model by choosing
$F=0$ and have been the subject of intense mathematical research under
different choices for $\phi$. The reader is referred to \cite{BlancLewin-2015}
for a survey on the available mathematical results. The setting $F=0$
corresponds  indeed to
the so-called Born-Oppenheimer approximation \cite{CondensMatter},
which is well adapted to the case of
very low temperatures and is based on the subsequent solution of the
electronic and the atomic problem. As mentioned, this approximation turns out to be
not always appropriate for metallic systems at finite temperatures
\cite{SuttonChen,LBMorse} and one is asked to tame the 
quantum nature of the problem.  This is however very
challenging from the mathematical viewpoint and rigorous optimality results for
point configurations in the quantum setting are scarce  \cite{BlancLebris,Betermin:2014fy}.
The EAM model represents hence an intermediate model between
zero-temperature phenomenological pair-interaction energies and
quantum systems. Electronic effects are still determined by atomic
positions, but in a more realistic nonlocal fashion  when $F$ is nonlinear, \EEE resulting in truly multi-body
interaction  systems, see \cite{DawBaskes83,FinnisSinclair,SuttonChen}
and \cite{DawFoilesBaskeReviewEAM} for a review.

The aim of this paper is to investigate point configurations
minimizing the EAM
energy. Being interested in periodic arrangements, we restrict our
analysis to the class of  {\it lattices},  namely 
infinite 
configurations of the form $L= \oplus_{i=1}^d \Z u_i$ where
 $\{u_i\}_{i=1}^d$  is a basis of $\R^d$. This reduces the optimality problem to finite
dimensions, making it analytically and numerically amenable. In
particular, the EAM energy-per-atom of the lattice $L$ takes the
specific form 
$$\mathcal E[L] =  F \Big(\sum_{q\in L\setminus \{0\}} \rho(|q|) \Big) + \sum_{q\in
  L\setminus \{0\}} \phi(|q|).  $$ 

In the classical pair-interaction case $F=0$, the  lattice  energy $\mathcal E$ has already received attention and a variety of
results are available, see
\cite{Mont,SarStromb,CohnKumar,BetTheta15,OptinonCM,BDefects20} and
the references therein. Such results are of course dependent on the
choice of the potential $\phi$. Three reference choices for $\phi$ are 
the {\it Gaussian} 
$\phi(r) = e^{-\pi\delta r^2}$ for $\delta>0$, the {\it inverse-power law}
$\phi(r) = r^{-s} $ for $s>d$, and the {\it
  Lennard-Jones-type} form $\phi(r) = ar^{-\alpha} - b r^{-\beta}$ for
$d<\beta<\alpha$ and $a,\,b >0$.  In the Gaussian case, it has been
shown by Montgomery \cite{Mont} that, for all $\delta>0$, the
triangular lattice   of unit density  is the unique minimizer (up to
isometries) of $\mathcal E$ with $F=0$  among unit-density
lattices. The same can be checked for the the inverse-power-law case by a
Mellin-transform argument. More generally, the minimality of
the
triangular lattice   of unit density is conjectured by Cohn and Kumar in
\cite[Conjecture 9.4]{CohnKumar} to hold among all unit-density periodic
configurations. This fact is  called \textit{universal optimality} and
has been recently proved in dimension $8$ and $24$ 
for the lattice $\mathsf{E}_8$ and the Leech lattice $\Lambda_{24}$,
respectively \cite{CKMRV2Theta}.
In the
Lennard-Jones case, the minimality  in 2d  of the triangular lattice at fixed
density has been investigated in \cite{Betermin:2014fy,BetTheta15},
the  minimality  in  3d  of the  cubic lattice is
proved  in
\cite{Beterminlocal3d}, and more general properties in arbitrary
dimensions have been investigated  in \cite{OptinonCM}. A recap of the main
properties of the  Lennard-Jones case  is  presented in
 Subsection  \ref{sec:LJ}.  These play a relevant  role in our analysis.

In this paper, we focus on the general case $F\not=0$,  when $F$ is nonlinear. \EEE More precisely,
we discuss the reference cases of embedding functions $F$ of the form
$$F(r) = r^t \log (\gamma r) \quad \text{or} \quad F(r) = r^t$$
for $t, \, \gamma >0$. The first, logarithmic choice is the classical
one chosen to fit with the so-called Universal Binding Curve (see e.g., \cite{Rose}) and \EEE favoring a specific minimizing value $r_0>0$, see
\cite{BanerjeaSmithUniv,CaiYe}. The second, power-law form favors on
the contrary $r_0=0$ and allows for a particularly effective
computational approach. 
Let us mention that other choices for $F$ could be of interest. In
particular, the form  $F(r)=-c\sqrt{r}$, $c>0$,
 is related to the {\it Finnis-Sinclair} model \cite{FinnisSinclair}
 and is discussed in Remark \ref{rmk:Finnis}. Some of our theory holds
 for general functions $F$, provided that they are minimized at a sole
 value $r_0$. We call such functions of {\it one-well} type.

The electronic-cloud contribution function $ \ff \colon \R_+ \to
\R_+$ is assumed to be decreasing and integrable. We specifically focus
on the Gaussian and inverse-power law 
$$  \ff  (r) = e^{-\delta r^2} \quad \text{or} \quad  \ff  (r) = r^{-s}$$
for $\delta>0$ and $s>d$, discussed, e.g., in  \cite{ZhangHuJiang19}  and
\cite{DawBaskes83,FinnisSinclair,SuttonChen}, respectively.

As for the pair-interaction potential $\phi\colon\R_+ \to
\R_+$, we assume a Lennard-Jones-type form
\cite{BaskesLJEAM,SrinivasanBaskesLJEAM} or an inverse-power law
 \cite{DawBaskes83,SuttonChen},  i.e., 
$$ \phi(r) = ar^{-\alpha} - b r^{-\beta} \quad \text{or} \quad \phi(r) = r^{-\alpha}$$
for $d<\beta<\alpha$ and $a,\,b>0$. Note that short-ranged potentials
$\phi$ have been considered  as well \cite{DawBaskes83,DawBaskes84}.

Our main theoretical results amount at identifying  minimizers in the
specific reference case of $F(r)=r \log r$ and $ \ff 
(r)=r^{-s}$.  More precisely, we find the following:
\begin{itemize}
\item (Inverse-power law) If $\phi(r)=r^{-\alpha}$, the minimizers of $\mathcal E$ coincide with
  those of the Lennard-Jones potential $r\mapsto r^{-\alpha} -
  r^{-s}$, up to rescaling (Theorem \ref{thmLJGeneral}); \\ 
\item (Lennard-Jones)  If  $\phi(r) = ar^{-\alpha} - b r^{-\beta}$,  under some compatibility  assumptions  on the
  parameters, the minimizers of $\mathcal E$ coincide with
  those of the Lennard-Jones potential $r\mapsto r^{-\alpha} -
  r^{-s}$  (Theorem~\ref{thmLJEAM2}). 
\end{itemize}
Actually, both results hold  for more general embedding
functions $F$, see \eqref{eq: g}  and Remarks \ref{rmk:applyrlogr}--\ref{rem:4}. With this at hand,  the problem can
be reduced to the pure Lennard-Jones case (i.e., $F =0$) which is
already well understood.    In particular, in the  two
dimensional  case we find that the   triangular lattice, up to rescaling  and isometries, is  the unique minimizer of $\mathcal E$  in  specific parameters regimes.  
These theoretical findings are illustrated by numerical
experiments in two and three dimensions. By alternatively choosing the
Gaussian $ \ff (r)=e^{-\delta r^2}$, in two dimensions we
additionally observe the onset of a phase transition between the
triangular and an orthorhombic lattice, as $\delta$ decreases. In
three dimensions, both in the inverse-power-law case $ \ff  (r) = r^{-s}$ and
in the Gaussian case $ \ff  (r)=e^{-\delta r^2}$, the simple cubic lattice
$\Z^3$ is favored against the  face-centered and the body-centered cubic lattice for $s$ or
$\delta$ small, respectively.

In the power-law case $F(r) = r^t$, for $ \ff $  of
inverse-power-law type  and
$\phi$ of Lennard-Jones type and specific, physically relevant choices of parameters, one
can conveniently reduce the complexity of the optimization
problem from the analytical standpoint. This reduction allows to
explicitly compute the EAM energy for any lattice of unit density,
hence allowing to investigate numerically  minimality in two and three
dimensions. Depending on the parameters, the relative minimality of
the triangular, square, and orthorhombic lattices in two dimensions and
the simple cubic, body-centered cubic, and face-centered cubic
lattices in three dimension is ascertained.

This is the plan of the paper: Notation on potentials and energies are
introduced in Subsections \ref{sec:Lattices} and
\ref{sec:PotEnergy}. The two subcases $F=0$ and $\phi=0$ are discussed
in Subsection \ref{sec:LJ} and  in Section~\ref{sec:EFEmbedding},
respectively. In particular, known results on Lennard-Jones-type interactions are
recalled in Subsection \ref{sec:LJ}. The inverse-power-law case
$\phi(r)=r^{-\alpha}$ is
investigated in Section \ref{sec:TheoricIP}. The Lennard-Jones
case  $\phi(r) = ar^{-\alpha} - b r^{-\beta}$ is addressed
theoretically and numerically in Section
\ref{sec:LJEAM}. In particular, Subsection \ref{sec:LJEAMclassic} contains the classical case $F(r) = r\log r$, and   Subsection \ref{sec:LJEAMPowerLaw} discusses the
power-law case $F(r)=r^t$.

\section{Notation and preliminaries}

\subsection{Lattices}\label{sec:Lattices}
 For any dimension $d$,  we write $\mathcal{L}_d$ for the set of all lattices $L=\bigoplus_{i=1}^d \Z u_i$, where $\{u_i\}_{i=1}^d$ is a basis of $\R^d$. We write $\mathcal{L}_d(1)\subset \mathcal{L}_d$ for the set of all lattices with unit density, which  corresponds to $|\det (u_1,\ldots,u_d)|= 1$. 

In dimension  two,  any lattice  $ L \in  \mathcal{L}_2(1)$ can be  written as 
\begin{equation*}
L:=\Z\left(\frac{1}{\sqrt{y}},0  \right)\oplus \Z\left( \frac{x}{\sqrt{y}},\sqrt{y}\right),
\end{equation*}
for $(x,y) \in \mathcal{D}$, where
\begin{equation}\label{eq:D}
\mathcal{D}=\big\{(x,y)\in \R^2 \, : \,  0\leq x \leq 1/2, \,  y>0; \,  x^2+y^2\geq 1\big\}
\end{equation}
is  the so-called {\it (half) fundamental} domain for $\mathcal{L}_2(1)$ (see, e.g., \cite[Page 76]{Mont}). In particular, the square lattice $\Z^2$ and the triangular lattice  with  unit density, denoted by  $\mathsf{A}_2 \in \mathcal{L}_d(1)$, are given by the respective choices $(x,y)=(0,1)$ and  $(x,y)=\left( 1/2,{\sqrt{3}}/{2}\right)$, i.e.,
\begin{equation*}
 \Z^2 =  \Z(1,0) \oplus \Z(0,1) \quad \text{ and } \quad  \mathsf{A}_2:=\sqrt{\frac{2}{\sqrt{3}}}\left[ \Z(1,0)\oplus \Z\left(\frac{1}{2},\frac{\sqrt{3}}{2} \right) \right].
\end{equation*}

In dimension three, the fundamental domain of $\mathcal{L}_3(1)$ is much more difficult to describe (see e.g., \cite[Section 1.4.3]{Terras_1988}) and its 5-dimensional nature makes it impossible to plot compared to the 2-dimensional  $\mathcal{D}$ defined in \eqref{eq:D}. The Face-Centered Cubic (FCC) and Body-Centered
Cubic (BCC) lattices with unit density are respectively indicated
by $\mathsf{D}_3\in \mathcal{L}_3(1)$ and $\mathsf{D}_3^*\in
\mathcal{L}_3(1)$, and are defined as 
\begin{align*}
 &\mathsf{D}_3:=2^{-\frac{1}{3}}\left[\Z(1,0,1)\oplus \Z(0,1,1)\oplus \Z(1,1,0)  \right];\\
 &\mathsf{D}_3^*:=2^{\frac{1}{3}}\left[\Z(1,0,0)\oplus \Z(0,1,0)\oplus \Z\left(\frac{1}{2},\frac{1}{2},\frac{1}{2}  \right)  \right]. 
\end{align*}

{\begin{remark}[Periodic configurations]
 All  results   in  this paper are stated in terms of
lattices, for the sake of definiteness. Let us however point out that the same statements hold
in the more general setting of periodic configurations in
dimensions $d\in \{8,24\}$,  on the basis of the recently proved
 optimality results from  \cite{CKMRV2Theta}. In dimension
$d=2$,  universal optimality is only known  among lattices,
see \cite{Mont}. Still, the validity of the Cohn-Kumar
conjecture (see \cite[Conjecture 9.4]{CohnKumar})  would allow us to
consider more general periodic configurations as well.  
\end{remark}

\subsection{Potentials and energies}\label{sec:PotEnergy}
For any dimension $d$, let $\mathcal{S}_d$ be the set of
all functions $ f  \colon\R_+\to \R$ such that $| f  (r)|={\rm
  O}(r^{-d-\eta})$ for some $\eta>0$ as $r\to \infty$.  By $\mathcal{S}_d^+ \subset \mathcal{S}_d$ we denote the subset of nonnegative functions. We say that  a
continuous function  $F\colon\R_+\to \R$ is a {\it one-well potential}
if there exists $r_0>0$ such that $F$ is  decreasing on $(0, r_0)$ and
increasing on $(r_0,\infty)$.

For any $\phi\in \mathcal{S}_d$, we define the {\it interaction energy} $E_\phi\colon\mathcal{L}_d\to \R$ by
\begin{equation}
E_\phi[L]:=\sum_{q\in L\backslash \{0\}} \phi(|q|).\label{fi}
\end{equation}
If $\phi(r)=r^{-s}$, $s>d$, $E_\phi[L]$ actually corresponds to
 the {\it Epstein zeta function}, which  is defined by
\begin{align}\label{eq: zeta}
\zeta_L(s):=\sum_{q\in L\backslash \{0\}} \frac{1}{|q|^s}.
\end{align}
For any function $F\colon  \R_+  \to \R$ and  for  any  $ \ff \in
\mathcal{S}_d^+$,  we  define the {\it embedding energy}  $E_{F,\ff}\colon\mathcal{L}_d\to \R$ by 
\begin{align}\label{eq: embedding -energy}
E_{F,\ff}[L]:=F(E_{ \ff }[L]) \quad \text{with} \quad
  E_{ \ff }[L]:=\sum_{q\in L\backslash \{0\}} \ff(|q|). 
\end{align}
Finally, for any $\phi\in \mathcal{S}_d$,  any $\ff \in \mathcal{S}_d^+$,  and any $F\colon\R_+  \to \R$, we define the  \emph{total energy}  $\mathcal{E}\colon\mathcal{L}_d\to \R$ by
\begin{align}\label{eq: main eenergy}
\mathcal{E}[L]:=E_{F,\ff}[L]+ E_\phi[L]=F(E_{ \ff }[L])+E_\phi[L].
\end{align}

 In the following, we investigate $\mathcal E$ under different
choices of the  potentials $F$, $ \ff $, and $\phi$.  In some parts, we will require merely abstract conditions on the potentials, such as a monotone decreasing $ \ff $ or a one-well potential $F$. In other parts, we will consider more specific potentials. In particular,  we will choose, for
$\gamma, \delta, t,a,b>0$, $s>d$, and $\alpha>\beta>d$,
$$
F(r)\in \{r^t,    r^{t} \log(\gamma r) \},\quad  \ff  (r)\in \{ r^{-s}, e^{-\delta r^2}\},\quad \phi(r)\in \{ r^{-\alpha}, a r^{-\alpha}- b r^{-\beta}\}.
$$
Note that the choice of $s$, $\delta$, $\alpha$, and $\beta$ implies
that   $\phi \in \mathcal{S}_d$  and $\ff \in \mathcal{S}_d^+$,  so that the sums
in \eqref{fi} and \eqref{eq: embedding -energy} are well defined.

For any $L\in \mathcal{L}_d(1)$, any $\phi\in \mathcal{S}_d$,   any $\ff \in \mathcal{S}_d^+$, 
and any $F\colon \R_+  \to \R$, we define, if they uniquely exist, the following optimal scaling parameters for the energies:
\begin{align}\label{eq:optimal lambda}
\lambda^{\mathcal{E}}_L  :=\argmin_{\lambda>0} \mathcal{E}[\lambda L],\quad  \lambda^{F,\ff}_L  :=\argmin_{\lambda>0} E_{F,\ff} [\lambda L],\quad  \lambda^\phi_L  :=\argmin_{\lambda>0} E_\phi[\lambda L].
\end{align}



\subsection{A recap on the Lennard-Jones-type energy}\label{sec:LJ}

A classical problem is to study the $F=0$ case for a  Lennard-Jones-type  potential
\begin{equation}\label{eq:LJ}
\phi(r)=a r^{-\alpha} -b r^{-\beta},\quad \alpha>\beta>d,\quad a,\,b>0.
\end{equation}
 Let us recap some known facts in this case
\cite{BetTheta15,OptinonCM}, which will be used later on.  We start by reducing the minimization problem on \emph{all} lattices  to  
a minimization problem  on lattices of \emph{unit density only}.  This is achieved by computing the  optimal scaling parameter of the energy $\lambda^\phi_L$ , see \eqref{eq:optimal lambda},  for each $L \in
\mathcal{L}_d(1)$, which in turn allows to find the minimum of the energy among dilations of $L$. More precisely, in case
\eqref{eq:LJ},  for all $\lambda>0$ and all  lattices  $L \in
\mathcal{L}_d(1)$, 
one has 
$$
E_\phi[\lambda L]=a\lambda^{-\alpha} \zeta_L(\alpha)-b\lambda^{-\beta} \zeta_L(\beta),
$$
where we use \eqref{eq: zeta}.  (This  energy was studied first in
\cite[Section 6.3]{BetTheta15}.) Then, we find the unique minimizer 
\begin{align}\label{eq: lambda2-def}
\lambda^\phi_L=  \left(\frac{\alpha a\zeta_L(\alpha)}{\beta b\zeta_L(\beta)}  \right)^{\frac{1}{\alpha-\beta}},
\end{align}
and therefore the energy is given by 
$$
\min_{\lambda>0} E_\phi[\lambda L] =  E_\phi[ \lambda^\phi_L L]=\frac{b^{\frac{\alpha}{\alpha-\beta}}\zeta_L(\beta)^{\frac{\alpha}{\alpha-\beta}}}{a^\frac{\beta}{\alpha-\beta}\zeta_L(\alpha)^\frac{\beta}{\alpha-\beta}}\left( \left(\frac{\beta}{\alpha}  \right)^{\frac{\alpha}{\alpha-\beta}}-\left(\frac{\beta}{\alpha}  \right)^{\frac{\beta}{\alpha-\beta}} \right)<0.
$$
The latter inequality follows from the fact that $\alpha > \beta$.  Consequently,  for any lattices $L,\Lambda\in \mathcal{L}_d(1)$, we have that
$$
E_\phi[\lambda_L^\phi L]\leq E_\phi[ \lambda^\phi_\Lambda \Lambda]\iff \frac{\zeta_L(\alpha)^\beta}{\zeta_L(\beta)^\alpha}\leq \frac{\zeta_\Lambda(\alpha)^\beta}{\zeta_\Lambda(\beta)^\alpha}.
$$
 This means that finding the lattice with minimal energy amounts to minimizing the function 
\begin{align}\label{eq: e*}
L \mapsto e^*(L) := \frac{\zeta_L(\alpha)^\beta}{\zeta_L(\beta)^\alpha}
\end{align}
 on  $\mathcal{L}_d(1)$. This is particularly effective in dimension two where for fixed
$(\alpha,\beta)$ the minimizer can be found numerically by plotting $L\mapsto
\min_\lambda \mathcal{E}[\lambda L]$ in the fundamental domain
$\mathcal{D}$.  Figure \ref{fig:LJ126} shows the  case 
$(\alpha,\beta)=(12,6)$, i.e., when $\phi$ is the classical
Lennard-Jones potential. The global minimum of $E_\phi$ in
$\mathcal{L}_2$ appears to be the triangular lattice
 $\lambda^\phi_{\mathsf{A}_2} \mathsf{A}_2$. 

For a certain range of parameters $(\alpha,\beta)$, this observation
can be rigorously ascertained.   
Indeed, for $d=2$, it is shown   in \cite[Theorem 1.2.B.]{BetTheta15}
that the global minimum of $E_\phi$ is uniquely achieved   by   a triangular lattice $\lambda_{\mathsf{A}_2}^\phi\mathsf{A}_2$ if  
\begin{equation}\label{eq:h}
 H  (\alpha)<  H (\beta),\quad \textnormal{where}\quad  H(t):=\frac{1}{2}\pi^{-t/2}\Gamma\left( \frac{t}{2}\right)t,
\end{equation}
 and $\Gamma$ is the classical {\it Gamma} function $\Gamma(r) =
\int_0^\infty x^{r-1}e^x \, {\rm d} x$ for $r>0$.  (In the sequel, all statements on uniqueness are intended up to isometries, without further notice.)   In fact, under
condition \eqref{eq:h} one has that \cite{BetTheta15} 
\begin{itemize}
\item $\mathsf{A}_2$ is the unique minimizer in $\mathcal{L}_2(1)$ of $\displaystyle L\mapsto  \lambda^\phi_L  =  \left(\frac{\alpha a\zeta_L(\alpha)}{\beta b\zeta_L(\beta)}  \right)^{\frac{1}{\alpha-\beta}}$,
\item $\mathsf{A}_2$ is the unique minimizer in $\mathcal{L}_2(1)$  of $e^*$ defined in \eqref{eq: e*}.
\end{itemize}
As pointed out in \cite[Remark  6.18]{BetTheta15}, it is necessary to choose
   $2<\beta<\alpha< M\approx 9.2045818$  in order
  to obtain  these optimality results by using the method
  developed there. In particular, this means that the
  following pairs of integer exponents can be chosen:
  $(\alpha,\beta)\in \{(4,3) ; (5,3) ; (6,3) ; (5,4) ; (6,4)
  \}$.  Note that  the classical Lennard-Jones potential $(\alpha,\beta)=(12,6)$ is not covered by \cite[Theorem 1.2.B.]{BetTheta15}.

\begin{figure}[H]
\begin{center}
	\includegraphics[width=10cm]{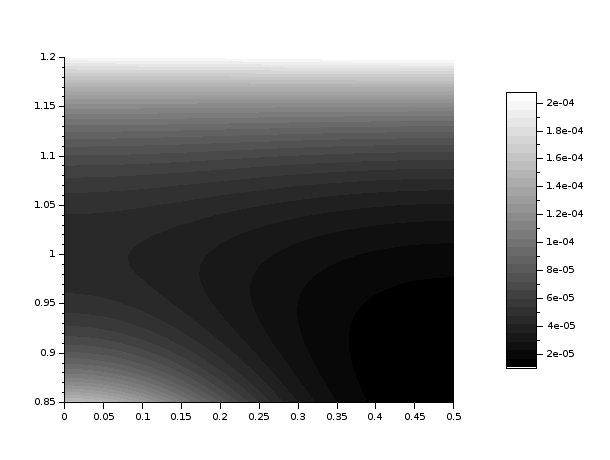}
	\end{center}
	\caption{Contour plot of $L\mapsto  e^*(L)  =
          \frac{\zeta_L(12)^{6}}{\zeta_L(6)^{12}}$ in the fundamental
          domain    $\mathcal D$.  The triangular lattice $\mathsf{A}_2$ with coordinates $(1/2,\sqrt{3}/2)$  appears to be the unique minimizer. Moreover, $\Z^2$ with coordinates $(0,1)$ appears to be a saddle point.}
	\label{fig:LJ126}
      \end{figure}

 We now ask ourselves
what is the minimal scaling parameter $\lambda$ and the
corresponding lattice  $L \in \mathcal{L}_d(1)$  for which $E_{\phi}[\lambda L]$ is minimized.
Physically, this would correspond to identifying the first 
minimum of $E_\phi$ starting from a high-density
configuration by progressively decreasing the density.  We have the following.

\begin{proposition}[Smallest volume meeting the global minimum]\label{lem:LJminlambda}
 Let $\phi$ be a Lennard-Jones-type potential  as in
 \eqref{eq:LJ}. If  $L_d \in \mathcal{L}_d(1)$ is the minimizer of
 $L\mapsto \zeta_L(\beta)$ on $\mathcal{L}_d(1)$ and
 $\lambda_{L_d}^\phi L_d$ is  the unique global minimizer of $E_\phi$ on $\mathcal{L}_d$, then $\lambda_{L_d}^\phi$ is the unique minimizer of $L\mapsto  \lambda^\phi_L$ on $\mathcal{L}_d(1)$.
 \end{proposition}

 \begin{proof}
As discussed above, if $\lambda_{L_d}^\phi L_d$ is a global minimizer of $E_\phi$ on $\mathcal{L}_d$,  then $L_d$ minimizes the function $e^*$  defined in \eqref{eq: e*} on $\mathcal{L}_d(1)$. This yields 
\begin{align}\label{eq:NNN}
\frac{\zeta_{L_d}(\alpha)^\beta}{\zeta_{L_d}(\beta)^\alpha} \leq \frac{\zeta_L(\alpha)^\beta}{\zeta_L(\beta)^\alpha}  
\end{align}
 for all $L\in \mathcal{L}_d(1)$. We thus have  
$$ \left( \frac{\zeta_L(\beta) \zeta_{L_d}(\alpha)}{\zeta_L(\alpha) \zeta_{L_d}(\beta)} \right)^\beta    = \frac{\zeta_L(\beta)^\alpha\zeta_{L_d}(\alpha)^\beta}{\zeta_L(\alpha)^\beta\zeta_{L_d}(\beta)^\alpha} \left( \frac{\zeta_{L_d}(\beta)}{\zeta_{L}(\beta)} \right)^{\alpha-\beta}      \leq \left( \frac{\zeta_{L_d}(\beta)}{\zeta_{L}(\beta)} \right)^{\alpha-\beta}.$$
As we are assuming  that  
$\zeta_{L_d}(\beta)\leq \zeta_L(\beta)$ for all $L\in \mathcal{L}_d(1)$,  we further get 
\begin{equation}\label{eq:ineqzeta1}
\frac{\zeta_L(\beta) \zeta_{L_d}(\alpha)}{\zeta_L(\alpha) \zeta_{L_d}(\beta)} \leq \left( \frac{\zeta_{L_d}(\beta)}{\zeta_{L}(\beta)} \right)^{\frac{\alpha-\beta}{\beta}}\leq 1,
\end{equation}
 where we use that $\alpha > \beta$. In view of \eqref{eq:
   lambda2-def},  this shows that  $\lambda^\phi_L \geq \lambda_{L_d}^\phi$ for all $L\in \mathcal{L}_d(1)$. If $\lambda^\phi_L = \lambda_{L_d}^\phi$, then we have a double equality in \eqref{eq:ineqzeta1}. This implies also equality in \eqref{eq:NNN} which is equivalent to  $e^*(L)=e^*(L_d)$. Therefore, it follows that $L=L_d$ up to rotation, by uniqueness of the minimizer $L_d$ of $e^*$. 
\end{proof}

\begin{figure}[H]
\begin{center}
	\includegraphics[width=10cm]{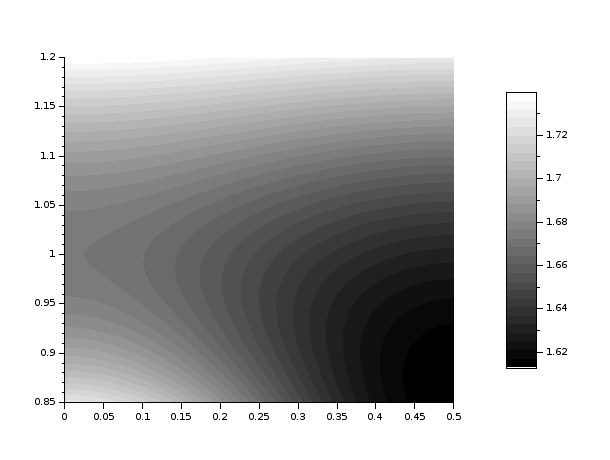}
	\end{center}
	\caption{Contour plot of $L\mapsto (b/a)^{1/6} \lambda_L^\phi=\left(\frac{12
              \zeta_L(12)}{6\zeta_L(6)}\right)^{1/6} $,
            see \eqref{eq: lambda2-def}, in the fundamental domain
             $\mathcal D$. The triangular lattice $\mathsf{A}_2$ with
             coordinates $(1/2,\sqrt{3}/2)$ is the unique minimizer
             {for any choice of $a,\, b>0$}.}
	\label{fig:LambdaLJ126}
\end{figure}

We refer to Figure \ref{fig:LambdaLJ126} for an illustration  in the  two-dimensional  case $(\alpha,\beta) = (12,6)$. Note that in this case the global minimum is not known. Still,  the triangular lattice appears to be the first stable structure reached by increasing the volume (decreasing the density). This is in agreement  with   Figure \ref{fig:LJ126} and Proposition~\ref{lem:LJminlambda}.  
Recall that the triangular lattice also minimizes $L\mapsto
\zeta_L(\beta)$ on $\mathcal{L}_d(1)$, as required in the statement of
Proposition~\ref{lem:LJminlambda}, see~\cite{Mont}.  \\

Notice that in dimension $d=3$ there is no rigorous result concerning the minimizer of $E_\phi$ in $\mathcal{L}_3$. Only local minimality results for cubic lattices $\Z^3,\mathsf{D}_3, \mathsf{D}_3^*$ have been derived in \cite{Beterminlocal3d}. Numerical investigations suggest that $\lambda_{\mathsf{D}_3}^\phi \mathsf{D}_3$ is the unique minimizer of $E_\phi$ in $\mathcal{L}_3$ for any values $\alpha>\beta>d$ of the exponents, see, e.g., \cite[Figure 5]{ModifMorse}, \cite[Figures 5 and 6]{OptinonCM} and \cite[Conjecture~1.7]{Beterminlocal3d}. Therefore, we can conjecture that $\mathsf{D}_3$ is the unique minimizer of $L\mapsto \lambda_L^\phi$ in $\mathcal{L}_3(1)$ by application of Proposition \ref{lem:LJminlambda}. 


\section{Properties of the embedding energy $E_{F,\ff}$}\label{sec:EFEmbedding}
 In this section we focus on the properties of the  embedding energy
 $E_{F,\ff}$ given in \eqref{eq: embedding -energy}. Although other choices for the potential $F$ may been considered
 (see, e.g., \cite{FinnisSinclair,DawFoilesBaskeReviewEAM}), we
 concentrate ourselves on the one-well case (see, e.g., \cite{CaiYe} and references therein). In
that case, it is  clear that the global minimum of
$E_{F,\ff}$ in $\mathcal{L}_d$ can be achieved for {\it any} $L$
 by simply choosing $\lambda$
such that $E_{ \ff }[\lambda L]= r_0 = \argmin_{r>0} F(r)$. We now ask ourselves
what is the minimal scaling parameter $\lambda$ and the
corresponding lattice  $L\in
\mathcal{L}_d(1)$  for which $E_{F,\ff}[\lambda L]$
achieves $\min F$. In other words,  what is the minimizer of
$L\mapsto \lambda^{F,\ff}_L$ in $\mathcal{L}_d(1)$  (recall \eqref{eq:optimal lambda}).   
Physically, this would correspond to reach the ground state
of the embedding energy $\min F$ starting from a high-density
configuration by progressively decreasing the density.

\begin{thm}[Smallest volume meeting the global minimum]\label{thm-generalFf}
Let $F\colon\R_+\to \R$ be a  one-well potential and let  $ \ff \in \mathcal{S}^+_d$ be
strictly decreasing.   Then, $\lambda^{F,\ff}_L$ exists and
$\min F$ is achieved by choosing $\lambda^{F,\ff}_L L$ for all $L\in
\mathcal{L}_d(1)$. Furthermore, if $L_d$ is the unique minimizer in
$\mathcal{L}_d(1)$ of $L\mapsto E_{ \ff }[ \lambda_{L_d}^{F,\ff} L]$, 
then $L_d$ is the unique minimizer in $\mathcal{L}_d(1)$ of $L\mapsto
\lambda^{F,\ff}_L$.  
\end{thm}
\begin{proof}
Let $r_0>0$ be the unique  minimizer of $F$, namely, $F(r_0)
= \min F$. Given any $L\in \mathcal{L}_d(1)$, the fact that  $ \ff
 \in \mathcal{S}_d^+$  is strictly decreasing implies that  $\lambda\mapsto
E_{ \ff }[\lambda L]$ is strictly decreasing  and   goes to $0$ at
infinity and to $\infty$ at $0$. Therefore, there exists a unique
$\lambda>0$ such that   $E_{ \ff }[\lambda L] = r_0$. 
Such $\lambda$ obviously  
coincides with $\lambda_L^{F,\ff}$ given in \eqref{eq:optimal lambda}. This shows the first part of the statement.  

 Suppose now  
that  $L_d$ is the unique minimizer in $\mathcal{L}_d(1)$ of
$L\mapsto E_{ \ff }[\lambda_{L_d}^{F,\ff} L]$. 
Assume  by contradiction that there exists $L \in \mathcal{L}_d(1)$, $L \neq L_d$, with $\lambda^{F,\ff}_L \le \lambda_{L_d}^{F,\ff}$. By using that $\lambda\mapsto
E_{ \ff }[\lambda L]$ is decreasing, this  would imply 
$$
E_{ \ff }[\lambda^{F,\ff}_L L]  \ge  E_{ \ff }[\lambda_{L_d}^{F,\ff} L]  >   E_{ \ff }[\lambda_{L_d}^{F,\ff} L_d]=  r_0  = E_{ \ff }[\lambda^{F,\ff}_L L],
$$
a contradiction.  
We thus  deduce that $\lambda_{L_d}^{F,\ff}\leq \lambda^{F,\ff}_L$ for all $L\in \mathcal{L}_d(1)$, with equality if and only if $L=L_d$.
\end{proof}

We note that Theorem \ref{thm-generalFf} can be applied to the
choice   $ \ff (r)=r^{-s}$, $s>d$, and  the triangular lattice $\mathsf{A}_2$, the
$\mathsf{E}_8$ lattice, or  the Leech lattice $\Lambda_{24}$ in
dimensions $2$, $8$, and $24$, respectively. In fact, these lattices are the unique minimizer of $L\mapsto E_{ \ff }[\lambda L]$  for all $\lambda>0$, see \cite{Mont,CKMRV2Theta}.

 Let us mention that, in this setting, asking $F$ to be one-well is
 not restrictive.  
 In fact, if $F$ is a strictly increasing (resp.\ decreasing)
 function, no optimal scaling parameters $\lambda>0$
  can be  found  since, for any $L \in \mathcal{L}_d(1)$,  $E_{F,\ff}[\lambda L]$ will be minimized for $\lambda \to 0$ (resp.\ $\lambda \to \infty$).

\section{The EAM energy with  inverse-power interaction $\phi(r)=r^{-\alpha}$}\label{sec:TheoricIP}

In this section, we study the energy $\mathcal{E}$ defined in \eqref{eq: main eenergy} when $\phi$ is given by the inverse-power interaction $\phi(r)=r^{-\alpha}$. The main result of this section is the following.

\begin{thm}[EAM energy for inverse-power interaction]\label{thmLJGeneral}
For any $\alpha>s>d$,  let $ \ff (r)=r^{-s}$,  let $\phi(r) =
r^{-\alpha}$, and let $F\in C^1(\R_+)$. We assume that the functions 
\begin{align}\label{eq: g}
g(r):=r^{1-{\alpha}/{s}}F'(r) \quad   \textnormal{and}\quad h(r):=F(r)-\frac{s}{\alpha}r F'(r) \quad \quad \text{for $r>0$}
\end{align}
satisfy  that $g$ is strictly increasing on $I:=  \lbrace F' < 0 \rbrace$, that $g(I) = (-\infty,0)$,  and that $h\circ g^{-1}$
is strictly decreasing on $(-\infty,0)$. (Note that $g^{-1}$ exists on
$(-\infty,0)$ and takes values in $\R_+$.)  Then,
 $\lambda^{\mathcal{E}}_{L}$ exists for all $L \in \mathcal{L}_d(1)$  and the following statements are equivalent:
\begin{itemize}
\item $L_d$ is the unique minimizer in $\mathcal{L}_d(1)$ of
$ L \mapsto 
e^*(L)=\frac{\zeta_L(\alpha)^s}{\zeta_L(s)^\alpha}$, see  \eqref{eq: e*}; 
\item $\lambda^{\mathcal{E}}_{L_d} L_d$ is the unique minimizer of $\mathcal{E}$ in $\mathcal{L}_d$;
\item  $\lambda_{L_d}^{ \bar{\phi}} L_d$ is the unique minimizer in $\mathcal{L}_d$ of  $
E_{\bar{\phi}}$  for $\bar{\phi}(r) = r^{-\alpha} - r^{-s}$, see \eqref{eq:LJ}. 
\end{itemize}
In particular, when $d=2$ and $ H   (\alpha)< H  (s)$
where $ H  $ is defined by \eqref{eq:h}, then the unique minimizer of
$\mathcal{E}$ in $\mathcal{L}_2$  is the triangular lattice
$\lambda_{\mathsf{A}_2}^\mathcal{E} \mathsf{A}_2$. 

Furthermore, if $L_d$ is the unique minimizer of $L\mapsto \zeta_L(s)$
in $\mathcal{L}_d(1)$ as well as a minimizer of   $e^*$  in $\mathcal{L}_d(1)$, then $L_d$ is the unique minimizer of $L\mapsto \lambda^{\mathcal{E}}_L$ in $\mathcal{L}_d(1)$, where $\lambda^{\mathcal{E}}_L$ is defined in \eqref{eq:optimal lambda}.
\end{thm}

The gist of this result is the coincidence of the minimizers of
$\mathcal E$ with those of  $E_{\bar{\phi}}$ for  $\bar{\phi}(r) =  r^{-\alpha} - r^{-s}$  (up to proper rescaling),  under
quite general choices of $F$. This
results in a simplification of the minimality problem for $\mathcal
E$ as one reduces to the study of minimality for the Lennard-Jones-type potential $\bar
\phi$, which is already well known, see  Subsection~\ref{sec:LJ}.  In particular, in two
dimensions and under condition $ H   (\alpha)< H  (s)$, the unique
minimizer is a properly rescaled triangular lattice.

Before proving the
theorem, let us present some  applications to specific choices of
$F$. 

\begin{remark}[Application 1 - The classical case $F(r)=r\log r$]\label{rmk:applyrlogr} 
 We can apply this theorem to $F(r)=r^t\log  (\gamma  r)$ for $t\in (0,\alpha/s)$  and $\gamma>0$  which is a one-well potential with minimum attained at point $r_0^t := \frac{1}{\gamma}
e^{-1/t}$. In particular, the case $F(r) = r\log r $ is admissible  since $s<\alpha$.  In fact, we have  $I= (0,r_0^t)$ and
\begin{align*}
  &g(r)=r^{t-\alpha/s}\big(t \log (\gamma r) +1),\quad
    g'(r)=r^{t-\alpha/s-1}\left(\left(t
    -\frac{\alpha}{s}\right)(t \log (\gamma r) +1) +t\right),\\
  &h(r)= r^t \left(\left(1-\frac{t s}{\alpha}\right) \log (\gamma r)
    -\frac{s}{\alpha}\right), \quad h'(r)= r^{t-1} \left(t\left(1-\frac{t s}{\alpha}\right) \log (\gamma r)
    -\frac{2t s}{\alpha} +1\right).
\end{align*}
Since  $g$ is strictly increasing on $(0,r_1^t)$ for
$r_1^t :=  \frac{1}{\gamma}  e^{\frac{2t s -\alpha}{t(\alpha - t
    s)}}$ and
$r_0^t < r_1^t$ we have that $g$ is strictly increasing
on $I$. Moreover,   $g(I)=(-\infty,0)$. On the other hand,  
$h$ is strictly decreasing on $(0,r_1^t)$.  Therefore, also $h\circ g^{-1}$ is strictly decreasing  on $(-\infty,0)$.  Hence, 
Theorem \ref{thmLJGeneral}  applies.   
\end{remark}


\begin{remark}[Application 2 - Finnis-Sinclair model]\label{rmk:Finnis}
Theorem \ref{thmLJGeneral} can also be applied to
$F(r)=-c\sqrt{r}$ for $c>0$. This case is known as the long-range {\it
  Finnis-Sinclair model} defined in \cite{SuttonChen}, based on the work of Finnis and Sinclair \cite{FinnisSinclair} on the description of cohesion in metals  and also used as a model to test the validity of machine-learning algorithms \cite{Hernandezetal}. In this  case, we obtain
$$
g(r)=-\frac{c}{2 r^{\frac{\alpha}{s}-\frac{1}{2}}}\quad \textnormal{and}\quad h(r)=c\sqrt{r}\left( \frac{s}{2\alpha}-1 \right).
$$
Since $s<\alpha<2\alpha$, $g$  is  strictly  increasing on $I=\{F' <  0\}=\R_+$, $g(I)=(-\infty,0)$, and $h$ is strictly decreasing on $\R_+$. Therefore, Theorem \ref{thmLJGeneral}  applies. 
\end{remark}

\begin{remark}[Application 3 - inverse-power law]
 Also the inverse-power law $F(r) =r^{-t}$ for $t
  >0$ satisfies the assumption of the theorem.  In fact, we have
 $$g(r) = -t r^{-t-{\alpha}/{s}} \quad \text{and}\quad
 h(r)=\left(1+ \frac{st}{\alpha}\right)r^{-t}.$$
 In particular, $g$ is strictly increasing on $I=\{F'<0\} = \R_+$ and
 $g(I)=(-\infty,0)$. Moreover,   $h$ is strictly decreasing on $\R_+$
 and therefore also $h\circ g^{-1}$  is strictly decreasing on $(-\infty,0)$. 
\end{remark}

\begin{remark}[Application 4 - negative-logarithm]\label{rem:4}
We can apply Theorem \ref{thmLJGeneral} to the inverse-logarithmic
case $F(r) =-\log r$. Indeed, we compute
 $$g(r) = -r^{-\alpha/s} \quad \text{and}\quad
 h(r)=   -\log r   +\frac{s}{\alpha}.$$
We hence have  that $g$ is strictly increasing on $I=\{F'<0\} = \R_+$ and
 $g(I)=(-\infty,0)$. As $h$ is strictly decreasing on $\R_+$, we have
 that $h\circ g^{-1}$  is strictly decreasing on $(-\infty,0)$.  
\end{remark}

\begin{proof}[Proof of Theorem \ref{thmLJGeneral}]
 In view of   \eqref{eq: zeta} and  \eqref{eq: main eenergy},  for any $\lambda>0$ and $L\in \mathcal{L}_d(1)$ we have that 
$$
\mathcal{E}[\lambda L]=F(\lambda^{-s} \zeta_L(s)) + \lambda^{-\alpha}\zeta_L(\alpha).
$$
The critical points of $\lambda \mapsto \mathcal{E}[\lambda L]$ for fixed $L$ are the solutions of
\begin{equation}\label{eq:criticptgen}
 \partial_\lambda \mathcal{E}[\lambda L] =  -    s\lambda^{-s-1}\zeta_L(s) \, F'(\lambda^{-s} \zeta_L(s)) - \alpha \lambda^{-\alpha-1}\zeta_L(\alpha)=0.
\end{equation}
This is equivalent  to 
$$ 
g(\lambda^{-s} \zeta_L(s))        =   -\frac{\alpha}{s} \frac{\zeta_L(\alpha)}{\zeta_L(s)^{\frac{\alpha}{s}}}   =     -\frac{\alpha}{s}e^*(L)^{\frac{1}{s}},
\quad 
$$
where $g$ is given in \eqref{eq: g}, and $e^*(L)= \frac{\zeta_L(\alpha)^s}{\zeta_L(s)^\alpha}$ was defined in \eqref{eq: e*}.   Since $g^{-1}$ is positive and  strictly  increasing on  $ (-\infty,0)$,  we have that the unique critical point is given by 
\begin{align}\label{eq: lambda*}
 \lambda^* :=   \left(  \frac{\zeta_L(s)}{g^{-1}\left(  -\frac{\alpha}{s}e^*(L)^{\frac{1}{s}} \right)}\right)^{\frac{1}{s}}.
\end{align}
 In view of  \eqref{eq:criticptgen},  we also have that $\partial_\lambda \mathcal{E}[\lambda L]\geq 0$ if
and only if $g(\lambda^{-s} \zeta_L(s))\leq
-\frac{\alpha}{s}e^*(L)^{\frac{1}{s}}$,  which is equivalent to $\lambda\geq \lambda^*$. 
 In particular, $\lambda\mapsto \mathcal{E}[\lambda L]$ is decreasing on $(0, \lambda^*) $ and  increasing on $ (\lambda^*, \infty)$.  This  shows that $\lambda^*$ is a minimizer and thus  $\lambda^* = \lambda^{\mathcal{E}}_L$, where $\lambda^{\mathcal{E}}_L$ is defined in \eqref{eq:optimal lambda}. 

 By using  the fact that
$(\lambda^{\mathcal{E}}_L)^{-\alpha}\zeta_L(\alpha)=-\frac{s}{\alpha}(\lambda^{\mathcal{E}}_L)^{-s}\zeta_L(s)
F'((\lambda^{\mathcal{E}}_L)^{-s}\zeta_L(s))$ from \eqref{eq:criticptgen}  and the identity $\lambda^* = \lambda^{\mathcal{E}}_L$,  the minimal
energy among dilated  copies $\lambda L$ of a given lattice
$L$ can be checked to be 
\begin{align*}
\mathcal{E}[\lambda^{\mathcal{E}}_L L]&=F\big((\lambda^{\mathcal{E}}_L)^{-s} \zeta_L(s)\big) + (\lambda^{\mathcal{E}}_L)^{-\alpha}\zeta_L(\alpha)\\
&=F\big((\lambda^{\mathcal{E}}_L)^{-s} \zeta_L(s)\big) -\frac{s}{\alpha}(\lambda^{\mathcal{E}}_L)^{-s}\zeta_L(s) F'\big((\lambda^{\mathcal{E}}_L)^{-s}\zeta_L(s)\big)\\
&=h\big((\lambda^{\mathcal{E}}_L)^{-s}\zeta_L(s)\big)\\
&= h\circ g^{-1}\left( -\frac{\alpha}{s}e^*(L)^{\frac{1}{s}}  \right),
\end{align*}
where $h$ is defined in \eqref{eq: g}. 
By assumption $h\circ g^{-1}$ is strictly decreasing on $(-\infty,0)$. 
Hence,  $L_d$ minimizes $L\mapsto \mathcal{E}[\lambda^{\mathcal{E}}_L
L]$ in $\mathcal{L}_d(1)$ (uniquely)  if and only if $L_d$ minimizes
$e^*$ (uniquely). This shows the equivalence of the first two
items in the statement. The equivalence to the third item has already been addressed in
the discussion before \eqref{eq: e*}. The two-dimensional case is a simple application of \cite[Theorem 1.2.B.]{BetTheta15} which ensures that  $\mathsf{A}_2$ is the unique minimizer of $e^*$ in $\mathcal{L}_2(1)$, as it has been already recalled in  Subsection~\ref{sec:LJ}. 

To complete the proof, it remains to show  the  final statement in $d$
dimensions. Assume that  $L_d$ is the unique minimizer of $L\mapsto \zeta_L(s)$
in $\mathcal{L}_d(1)$ as well as a minimizer of  $e^*$ in $\mathcal{L}_d(1)$. In this case, by using \eqref{eq: lambda*} and the identity $\lambda^* = \lambda^{\mathcal{E}}_L$, it indeed follows that  $L_d$ is   the unique minimizer of $L\mapsto \lambda^{\mathcal{E}}_L$ in $\mathcal{L}_d(1)$, since $g^{-1}$ is positive and increasing on  $(-\infty,0)$. 
\end{proof}

\section{The EAM energy with Lennard-Jones-type interaction
  $\phi(r) = a r^{-\alpha}-br^{-\beta}$} \label{sec:LJEAM}

We now move on to consider the  full EAM  energy $\mathcal{E}$
defined in \eqref{eq: main eenergy} for  Lennard-Jones-type   potentials $\phi$ as in \eqref{eq:LJ}. We split this
section into two parts.  At  first, we address  
the classical case $F(r)=r\log r$ analytically and numerically. Afterwards, we provide some further numerical studies  for the power law case $F(r) = r^t$.

\subsection{The classical case $F(r)=r\log r$}\label{sec:LJEAMclassic}
 We start with two theoretical results and then proceed with several numerical investigations.

\subsubsection{Two theoretical results}

The following corollary  is a  straightforward application of Theorem
\ref{thm-generalFf}.  



\begin{corollary}[Existence of parameters for the optimality of $\mathsf{A}_2$]\label{cor:A2opt}
Let 
$$
F(r)=r^{ t}\log (\gamma r), \quad \ff  (r)=r^{-s},  \quad \textnormal{and}\quad \phi(r)=a r^{-\alpha} - b r^{-\beta},
$$
 for $\gamma, t  >0$, $s>2$,  $\alpha>\beta>2$,  and $a,b>0$. Then, given parameters $(\alpha,\beta, \gamma, s,{t})$ such that $
   H  (\alpha)< H   (\beta)$,  where $H$ is defined in \eqref{eq:h}, 
one can find  coefficients $a$ and $b$  such that  the unique global minimizer in $\mathcal{L}_2$ of $\mathcal{E}$
is the triangular lattice $\lambda_{\mathsf{A}_2} \mathsf{A}_2$ where
$$\lambda_{\mathsf{A}_2}=e^{\frac{ t^{-1} +\log\gamma}{s}}\zeta_{\mathsf{A}_2}(s)^{\frac{1}{s}}.$$
 Moreover, $\mathsf{A}_2$ is the   unique minimizer of $L\mapsto \lambda^{\mathcal{E}}_L$ in $\mathcal{L}_2(1)$. 
\end{corollary}


\begin{proof}
We first remark that $F$ and $ \ff $ satisfy  the assumption of Theorem \ref{thm-generalFf}.  By recalling  \eqref{eq: zeta},  \eqref{eq:optimal lambda} and using the fact that $\argmin F = \frac{1}{\gamma} e^{-1/t}$,   we have 
$$
\lambda_{\mathsf{A}_2}^{F,\ff}=e^{\frac{ t^{-1} +\log\gamma}{s}}\zeta_{\mathsf{A}_2}(s)^{\frac{1}{s}},
$$ 
and $E_{F,\ff}[\lambda_{\mathsf{A}_2}^{F,\ff} \mathsf{A}_2] =  F( (\lambda_{\mathsf{A}_2}^{F,\ff})^{-s} \zeta_{\mathsf{A}_2}(s)   )   = \min F$.  
 On the other hand, we know from  \cite[Theorem 1.2]{BetTheta15} that $E_\phi$ is uniquely  minimized in $\mathcal{L}_2$ by $\lambda_{\mathsf{A}_2}^\phi \mathsf{A}_2$ where 
$$
\lambda_{\mathsf{A}_2}^\phi   =  \left( \frac{\alpha a \zeta_{\mathsf{A}_2}(\alpha)}{\beta b \zeta_{\mathsf{A}_2}(\beta)}\right)^{\frac{1}{\alpha-\beta}},
$$
see \eqref{eq: lambda2-def}.  Hence,   if
 $\lambda_{\mathsf{A}_2}^{F,\ff}=\lambda_{\mathsf{A}_2}^\phi$,  then   $\lambda_{\mathsf{A}_2}^{F,\ff} \mathsf{A}_2=
\lambda_{\mathsf{A}_2}^\phi \mathsf{A}_2$ is the unique minimizer of the
sum of the two energies $E_{F,\ff}$ and $E_\phi$.  The identity 
$\lambda_{\mathsf{A}_2}^{F,\ff}=\lambda_{\mathsf{A}_2}^\phi$  is  equivalent to equation 
$$
\frac{a}{b}=\frac{\beta \zeta_{\mathsf{A}_2}(\beta)}{\alpha \zeta_{\mathsf{A}_2}(\alpha)}\zeta_{\mathsf{A}_2}(s)^{\frac{\alpha-\beta}{s}}e^{\frac{\alpha-\beta}{s}(t^{-1}  +\log \gamma)}.
$$
 For this choice of $a$ and $b$, we thus get that the unique global minimizer in $\mathcal{L}_2$ of $\mathcal{E}$
is the triangular lattice $\lambda^\mathcal{E}_{\mathsf{A}_2} \mathsf{A}_2$ with   $\lambda^{\mathcal{E}}_{\mathsf{A}_2}=\lambda_{\mathsf{A}_2}^{F,\ff}=\lambda_{\mathsf{A}_2}^\phi$.
The last  statement follows by applying
Proposition  \ref{lem:LJminlambda} to $L_d=\mathsf{A}_2$. 
\end{proof}

The drawback of the result is that it is not \emph{generic} in the
sense that it holds only for specific coefficients $a$ and $b$. We now
give a result  which holds  in any dimension for \emph{all}
coefficients $a,\,b>0$, at the expense of the fact that
$\phi$ and $ \ff $ need to have the same decay  ${\rm
  O}(r^{-s})$. In this regard, the result  is in the  spirit  Theorem~\ref{thmLJGeneral}
 but under the choice   $\phi(r)=ar^{-\alpha}-br^{-s}$.

\begin{thm}[EAM energy for Lennard-Jones-type  interaction]\label{thmLJEAM2}
Let $F$ be  as in Theorem \ref{thmLJGeneral}  and additionally suppose that $F$ is convex and in  $C^2(\R_+)$. Let  
$$
 \ff  (r)=r^{-s},\quad \phi(r)=ar^{-\alpha}-br^{-s},\quad \text{ for } \    d  <s<\alpha  \quad \text{and} \quad a,\,b>0.
$$
Then,  $\lambda^{\mathcal{E}}_{L}$ exists for all $L \in \mathcal{L}_d(1)$  and the following statements are equivalent:
\begin{itemize}
\item $L_d$ is the unique minimizer of $ L \mapsto 
e^*(L)=\frac{\zeta_L(\alpha)^s}{\zeta_L(s)^\alpha}$, see \eqref{eq: e*}; 
\item $\lambda_{L_d}^{\mathcal{E}} L_d$ is the unique minimizer of $\mathcal{E}$ in $\mathcal{L}_d$;
\item $\lambda_{L_d}^{\phi} L_d$ is the unique minimizer in $\mathcal{L}_d$ of
 $
E_{{\phi}}$. 
\end{itemize}
In particular, when $d=2$ and $ H   (\alpha)< H  (s)$
where $ H  $ is defined by \eqref{eq:h}, then the unique minimizer of
$\mathcal{E}$ in $\mathcal{L}_2$ is the triangular lattice
$\lambda_{\mathsf{A}_2}^\mathcal{E} \mathsf{A}_2$.

Furthermore, if $L_d$ is the unique minimizer of $L\mapsto \zeta_L(s)$ in $\mathcal{L}_d(1)$ as well as a  minimizer  of $e^*$ in $\mathcal{L}_d(1)$, then $L_d$ is the unique minimizer of $L\mapsto \lambda^{\mathcal{E}}_L$ in $\mathcal{L}_d(1)$.
\end{thm}
\begin{proof}
 In view of \eqref{eq: zeta},  the energy $\mathcal E$ can be written as 
$$
\mathcal{E}[L]=F(\zeta_L(s))+a\zeta_L(\alpha)-b\zeta_L(s)=a\big(\tilde{F}(\zeta_L(s))+\zeta_L(\alpha)\big),
$$
where $\tilde{F}(r)=a^{-1}(F(r)-br)$.  In a similar fashion to \eqref{eq: g}, we define 
$$\tilde g(r) := r^{1-{\alpha}/{s}}\tilde F'(r) =a^{-1}g(r)-\frac{b}{a}r^{1-{\alpha}/{s}}
, \quad
\tilde h(r) := \tilde F(r) -\frac{s}{\alpha} r \tilde F'(r) = a^{-1}h(r)-\frac{b}{a}\Big(1-\frac{s}{\alpha}\Big)r,
$$
where $g$ and $h$ are defined in \eqref{eq: g}. 
 We first check that $\tilde{g}$  is strictly increasing on $\tilde{I}:= \lbrace  \tilde{F}' < 0 \rbrace$. Indeed, since  $F$ (and hence $\tilde F$) is convex and $\alpha >s$, we get that  
$$\tilde g'(r) = \left(1-\frac{\alpha}{s}\right)r^{-{\alpha}/{s}}\tilde F'(r) +
r^{1-{\alpha}/{s}}\tilde F''(r) \ge
\left(1-\frac{\alpha}{s}\right) r^{-{\alpha}/{s}}\tilde F'(r)>0$$
for all $r \in \tilde{I}$. Since by assumption $g(\lbrace F' < 0 \rbrace) = (-\infty,0)$ and $\tilde{I} = \lbrace  \tilde{F}' < 0 \rbrace \supset \lbrace {F}' < 0 \rbrace$, we find $\tilde{g}(\tilde{I}) = (-\infty,0)$. Eventually,  $\tilde{h}\circ \tilde{g}^{-1}$ is strictly decreasing on $(-\infty,0)$, as well.  We can hence apply 
Theorem~\ref{thmLJGeneral}  and obtain the assertion.  
\end{proof}

\begin{remark}
As a consequence of Remark \ref{rmk:applyrlogr}, the previous result can be applied to $F(r)=r\log r$. Already for this $F$, in the case of a more general Lennard-Jones potential $\phi(r)=ar^{-\alpha}-br^{-\beta}$, the equation for the critical points of $\lambda\mapsto \mathcal{E}[\lambda L]$ for a fixed lattice $L$ is
$$
\log \lambda = \frac{a'}{b'}\lambda^{s-\alpha}-\frac{d'}{b'}\lambda^{s-\beta} + \frac{c'}{b'}
$$
 for 
$a'=\alpha a \zeta_L(\alpha)$, $b' = s^2 \zeta_L(s)$, $c' =
s\zeta_L(s)(1+\log\zeta_L(s))$, and  $d'=\beta b \zeta_L(\beta)$. 
This  is generically not solvable in closed form  when $s\neq
\beta$, and  makes the computation of $\mathcal{E}[\lambda^{\mathcal{E}}_L L]$ more difficult. This is why we choose $s = \beta$ in the above result. 
\end{remark}

\subsubsection{Numerical investigation in 2d}


We choose $s$ as parameter and fix $t=\gamma=a=b=1$, and $\alpha=12$, $\beta = 6$,  i.e.,
\begin{equation}
F(r)=r\log r,\quad  \ff  (r)=r^{-s},\quad \phi(r)=\frac{1}{r^{12}}-\frac{1}{r^6}.\label{case0}
\end{equation}
 We employ here a gradient descent method, which is rather
computationally intensive. Note that a more efficient numerical method will be
amenable in  Subsection \ref{sec:LJEAMPowerLaw},  as an effect of a different structure of the potentials.
Numerically, we observe the following  (see Figure~\ref{fig:plotClassLJIPs}): 
\begin{itemize}
\item For $s>s_1$, $s_1\approx 5.14$, the triangular lattice
  $\lambda_{\mathsf{A}_2}^\mathcal{E}\mathsf{A}_2$ is apparently the unique global minimizer of $\mathcal{E}$.
\item For $s<s_1$, the energy does not seem to have a global minimizer.
\end{itemize}
Furthermore, for $s>s_0$, $s_0\approx 5.09$, we have checked (see Figure \ref{fig:plotTriSq}) that 
\begin{equation*}
\min_\lambda \mathcal{E}[\lambda \Z^2]= \mathcal{E}[\lambda_{\Z^2}^\mathcal{E}\Z^2]>\mathcal{E}[\lambda_{\mathsf{A}_2}^\mathcal{E}\mathsf{A}_2]=\min_\lambda \mathcal{E}[\lambda \mathsf{A}_2],
\end{equation*}
whereas the inequality is reversed if $s<s_0$.

\begin{figure}[H]
\begin{center}
	\includegraphics[width=79mm]{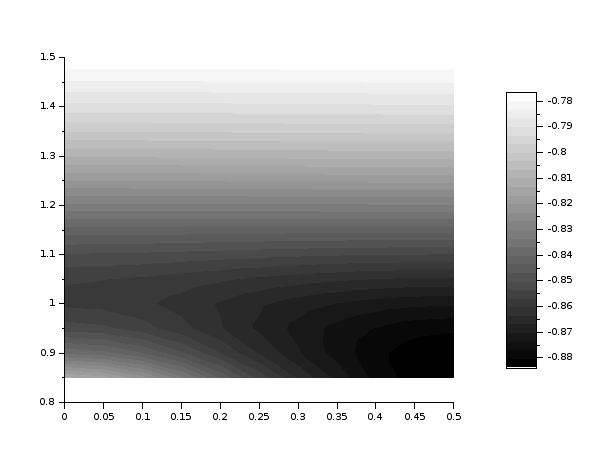}\quad 	\includegraphics[width=79mm]{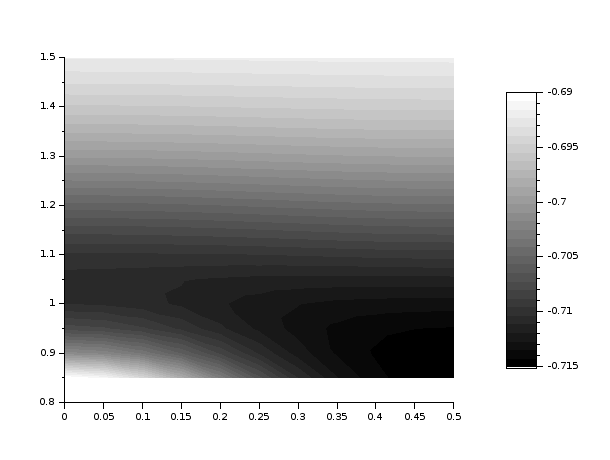}\\\includegraphics[width=79mm]{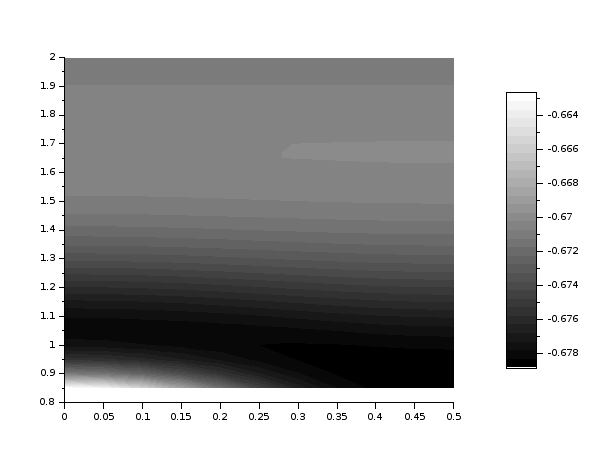} \quad \includegraphics[width=79mm]{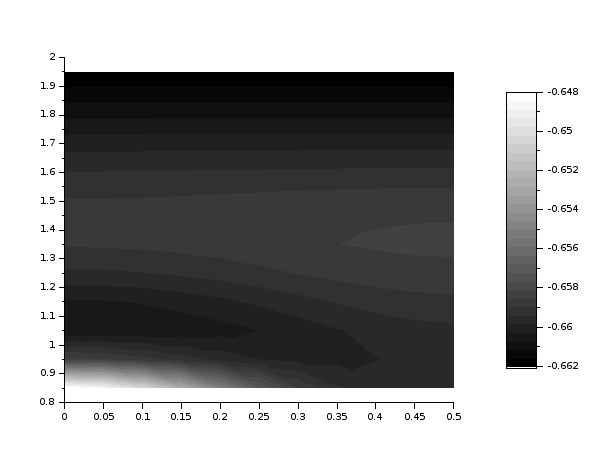}  
	\end{center}
	\caption{Case \eqref{case0} in two dimensions. Plot of $L\mapsto \min_\lambda \mathcal{E}[\lambda L]$ in the fundamental domain $\mathcal{D}$ for $s=6$ (up left), $s=5.3$ (up right), $s=5.15$ (down left) and $s=5.07$ (down right).}
	\label{fig:plotClassLJIPs}
\end{figure}

\begin{figure}[H]
\begin{center}
	\includegraphics[width=8cm]{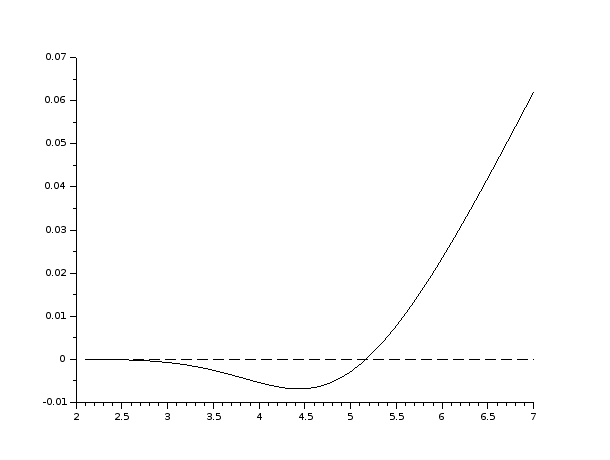}
	\end{center}
	\caption{Case \eqref{case0} in two dimensions. Plot of $s\mapsto \min_\lambda \mathcal{E}[\lambda \Z^2]- \min_\lambda \mathcal{E}[\lambda \mathsf{A}_2]$ for $s\in [2.1,7]$.}
	\label{fig:plotTriSq}
\end{figure}

We now replace $ \ff $ by a Gaussian function. Namely, we consider  the case
\begin{equation}
F(r)=r\log r,\quad  \ff  (r)=e^{-\delta r^2},\quad \phi(r)=\frac{1}{r^{12}}-\frac{1}{r^6}.\label{case1}
\end{equation}

\begin{figure}[H]
\begin{center}
	\includegraphics[width=8cm]{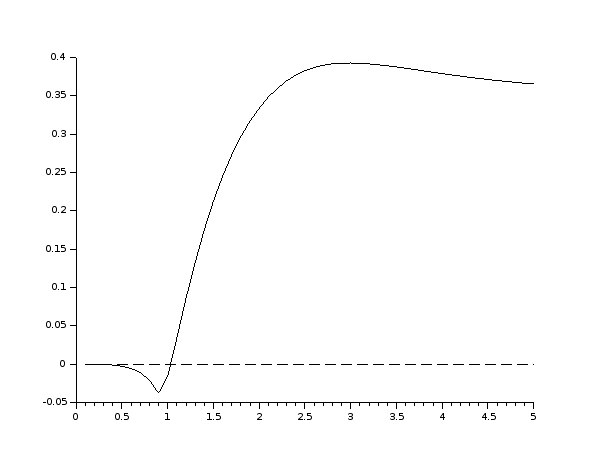}
	\end{center}
	\caption{Case \eqref{case1} in two dimensions. Plot of $ \delta\mapsto \min_\lambda \mathcal{E}[\lambda \Z^2]- \min_\lambda \mathcal{E}[\lambda \mathsf{A}_2]$ for $\delta\in [0.1,5]$.}
	\label{fig:plotTriSqGauss}
\end{figure}

In this case, the triangular lattice $\lambda_{\mathsf{A}_2}^\mathcal{E}\mathsf{A}_2$ still
seems to be minimizing $\mathcal{E}$ for large $\delta$,  see Figure \ref{fig:ClassLJGaussa2}.  More precisely:                                            
\begin{itemize}
\item There exists $\delta_0\approx 1.04$ such that, for $\delta>\delta_0$, the triangular lattice $\lambda_{\mathsf{A}_2}^\mathcal{E}\mathsf{A}_2$ is the global  minimizer  of $\mathcal{E}$ in $\mathcal{L}_2$.
\item For $\delta<\delta_0$, the global minimizer of $\mathcal{E}$ seems to move (continuously) in $\mathcal{D}$ increasingly following the $y$-axis as $\delta$ decreases to $0$. For instance, 
\begin{itemize}
\item If $\delta=1$, then the minimizer is $(0,y_1)$ where $y_1\approx 1.014$.
\item If $\delta=0.95$, then the minimizer is $(0,y_{0.95})$ where $y_{0.95}\approx 1.665$.
\end{itemize}
\item Furthermore, we have checked that, for $\delta>\delta_0$,
$$\min_\lambda \mathcal{E}[\lambda \Z^2]= \mathcal{E}[\lambda_{\Z^2}^\mathcal{E}\Z^2]>\mathcal{E}[\lambda_{\mathsf{A}_2}^\mathcal{E}\mathsf{A}_2]=\min_\lambda \mathcal{E}[\lambda \mathsf{A}_2],
$$
whereas the inequality is reversed if $\delta<\delta_0$ (see Figure \ref{fig:plotTriSqGauss}).
\end{itemize}

\begin{figure}[H]
\begin{center}
	\includegraphics[width=8cm]{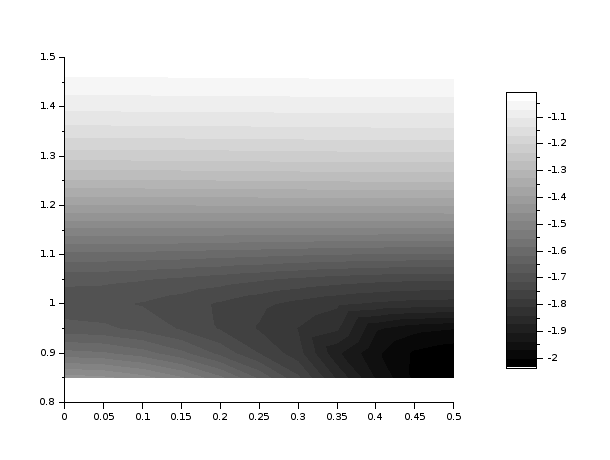}\quad \includegraphics[width=8cm]{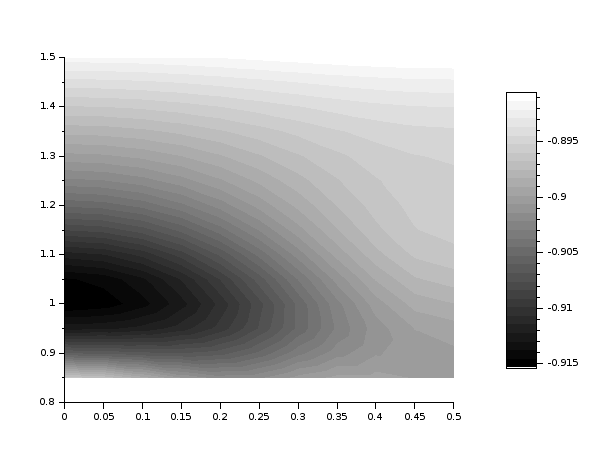}\\
	\includegraphics[width=8cm]{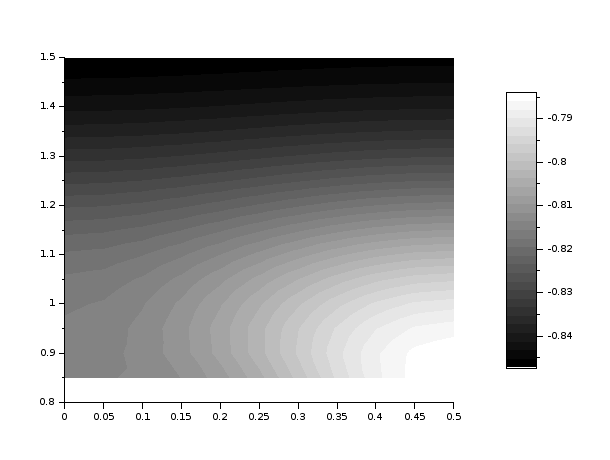}\quad \includegraphics[width=8cm]{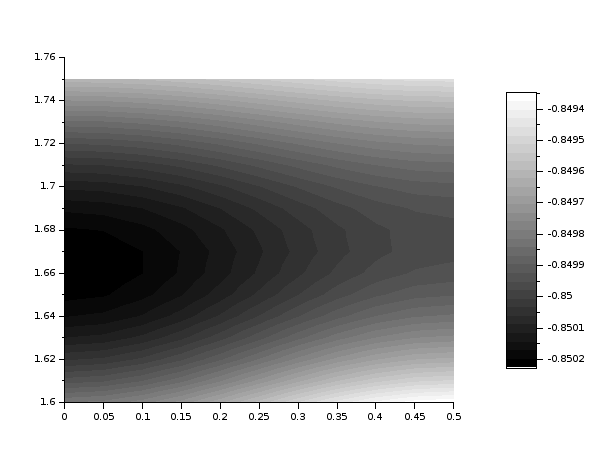}
	\end{center}
	\caption{Case \eqref{case1} in two dimensions. Plot of $L\mapsto \min_\lambda \mathcal{E}[\lambda L]$ when $\delta=2$ (up left), $\delta=1$ (up right) and $\delta=0.95$ (down left and right) in the fundamental domain $\mathcal{D}$.}
	\label{fig:ClassLJGaussa2}
\end{figure}

\subsubsection{Numerical investigation in 3d}
 Let us go back to case \eqref{case0}, now in three
dimensions. 
We investigate the difference of energies between the Simple
Cubic (SC), Face-Centered Cubic (FCC), and Body-Centered Cubic (BCC) lattices, namely,
$\Z^3, \mathsf{D}_3, \mathsf{D}_3^*$, as $s$ increases.  Examples
of FCC and BCC metals are Al, Cu, Ag, Au, Ni, Pd, Pt, and Nb, Cr, V,
Fe, respectively \cite{Wells}. Po is the only metal  crystallizing  in a SC
structure~\cite{Silva}. 

Before giving
 our numerical results, let us remark that   the lattices  $\Z^3$,
 $\mathsf{D}_3$, and $\mathsf{D}_3^*$  are critical points of
 $\mathcal{E}$ in $\mathcal{L}_3(1)$. Moreover, 
 recall the following conjectures:
\begin{itemize}
\item Sarnak-Strombergsson's conjecture (see \cite[Equation
  (44)]{SarStromb}): for all $s \ge  3/2$ (and in particular for
  $s>3$, so that   $r \mapsto r^{-s} \in \mathcal  S_3^+$),   $\mathsf{D}_3$ is the unique minimizer of $L\mapsto \zeta_L(s)$ in $\mathcal{L}_3(1)$. 
\item The global minimizer of the Lennard-Jones energy $E_\phi$  is $\lambda_{\mathsf{D}_3}^\phi \mathsf{D}_3$  (see e.g.~\cite[Figure~5]{ModifMorse} and \cite[Conjecture~1.7]{Beterminlocal3d}). 
\end{itemize}

We have numerically studied the following function
$$
s\mapsto \min_{\lambda >0} \mathcal{E}[\lambda L],\quad L\in \{\mathsf{D}_3,\mathsf{D}_3^*,\Z^3 \}
$$
 for $s>3$, see Figure~\ref{fig:plotIP3dclassic}.  We have found that there exist $s_0< s_1<s_2$ where $s_0\approx 5.4985$, $s_1\approx 5.576$, and $s_2\approx 5.584$ such that
\begin{itemize}
\item For $s\in (3,s_0)$, $ \min_{\lambda>0} \mathcal{E}[\lambda \Z^3]<\min_{\lambda>0} \mathcal{E}[\lambda \mathsf{D}_3^*]<\min_{\lambda >0} \mathcal{E}[\lambda  \mathsf{D}_3]$;
\item For $s\in (s_0,s_1)$, $\min_{\lambda >0} \mathcal{E}[\lambda \Z^3]<\min_{\lambda >0} \mathcal{E}[\lambda \mathsf{D}_3]<\min_{\lambda >0} \mathcal{E}[\lambda \mathsf{D}_3^*]$;
\item For $s\in (s_1,s_2)$, $\min_{\lambda >0} \mathcal{E}[\lambda \mathsf{D}_3]<\min_{\lambda >0} \mathcal{E}[\lambda \Z^3]<\min_{\lambda >0} \mathcal{E}[\lambda \mathsf{D}_3^*]$;
\item For $s>s_2$, $ \min_{\lambda >0} \mathcal{E}[\lambda \mathsf{D}_3]<\min_{\lambda >0} \mathcal{E}[\lambda \mathsf{D}_3^*]<\min_{\lambda>0} \mathcal{E}[\lambda  \Z^3]$.
\end{itemize}

 It is remarkable that  for small values of $s$ the simple cubic lattice $\Z^3$ has lower
energy with respect to the usually energetically favored
$\mathsf{D}_3$ and $\mathsf{D}_3^*$. 

\begin{figure}[H]
\begin{center}
	\includegraphics[width=8cm]{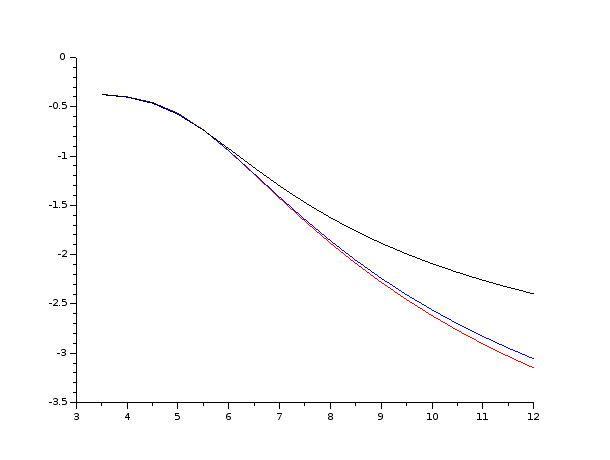}\quad \includegraphics[width=8cm]{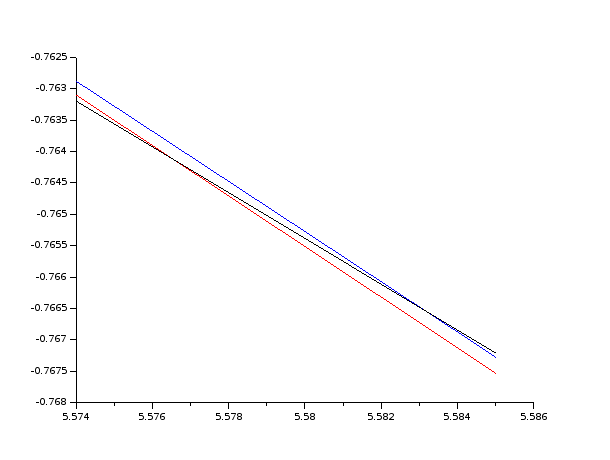}
	\end{center}
	\caption{Case \eqref{case0} in three dimensions.  Plots of $s\mapsto \min_\lambda \mathcal{E}[\lambda L]$ for $L=\mathsf{D}_3$ (red), $L=\mathsf{D}_3^*$ (blue) and $L=\Z^3$ (black) on  two different intervals.}
	\label{fig:plotIP3dclassic}
\end{figure}

 Consider now the Gaussian case \eqref{case1} in three dimensions. 
The total energy then reads 
$$
\mathcal{E}[L]:=\theta_L(\delta) \log \theta_L(\delta) + \zeta_L(12)-\zeta_L(6), \quad \textnormal{where}\quad \theta_L(\delta):=\sum_{p\in L\backslash \{0\}} e^{-\delta |p|^2}.
$$
 In the following, we will call $\theta_L(\delta)$ the {\it lattice theta function} with parameter
$\delta>0$. Note however that under this name one usually refers to
such sum including the term for $p=0$ and with weight $e^{-\delta \pi
  |p|^2}$. 

We recall the following conjectures: 
\begin{itemize}
\item  Sarnak-Strombergsson's  conjecture (see \cite[Equation (43)]{SarStromb}): if $\delta<\pi$, then $\mathsf{D}_3^*$ minimizes $L\mapsto \theta_L(\delta)$ in $\mathcal{L}_3(1)$. If $\delta>\pi$, then $\mathsf{D}_3$ minimizes the same lattice theta function  in $\mathcal{L}_3(1)$ (with a coexistence phase around $\pi$ actually).
\item As mentioned before,  the unique minimizer of the Lennard-Jones energy $E_\phi$ in $\mathcal{L}_3$ is $\lambda_{\mathsf{D}_3}^\phi \mathsf{D}_3$ (see e.g.~\cite{Beterminlocal3d} and \cite[Figure~5]{ModifMorse}). 
\end{itemize}


In Figure \ref{fig:3dLJTheta} we plot the functions $\delta\mapsto \min_{\lambda>0} \mathcal{E}[\lambda L]$  for  $L\in \{\mathsf{D}_3,\mathsf{D}_3^*,\Z^3 \}$. We numerically observe that  there  exist   $0<\delta_1< \delta_2<\delta_3$, where $\delta_1\approx 1.13$, $\delta_2\approx 1.21$, and $\delta_3\approx 1.223$ such  that
\begin{itemize}
\item for all $\delta\in (0, \delta_1)$, $ \min_{\lambda>0} \mathcal{E}[\lambda \Z^3]<\min_{\lambda>0} \mathcal{E}[\lambda \mathsf{D}_3^*]<\min_{\lambda >0} \mathcal{E}[\lambda \mathsf{D}_3]$; 
\item for all $\delta\in (\delta_1,\delta_2)$, $ \min_{\lambda>0} \mathcal{E}[\lambda \Z^3]<\min_{\lambda>0} \mathcal{E}[\lambda \mathsf{D}_3]<\min_{\lambda>0} \mathcal{E}[\lambda \mathsf{D}_3^*]$;
\item for all $\delta\in (\delta_2,\delta_3)$, $ \min_{\lambda >0} \mathcal{E}[\lambda \mathsf{D}_3]<\min_{\lambda>0} \mathcal{E}[\lambda \Z^3]<\min_{\lambda>0} \mathcal{E}[\lambda \mathsf{D}_3^*]$;
\item for all $\delta>\delta_3$,  $ \min_{\lambda>0} \mathcal{E}[\lambda \mathsf{D}_3]<\min_{\lambda>0} \mathcal{E}[\lambda \mathsf{D}_3^*]<\min_{\lambda >0} \mathcal{E}[\lambda \Z^3]$.
\end{itemize}

It is indeed important 
that  the EAM  energy 
favors $\mathsf{D}_3$ or $\mathsf{D}_3^*$ for  some
specific 
choice of parameters.  In fact, FCC and BCC lattices are commonly
emerging in metals. 
It is also remarkable that the 
simple cubic lattice $\Z^3$ (up to
rescaling)  is favored with respect to $\mathsf{D}_3$ or
$\mathsf{D}_3^*$ for some other  choice of parameters. In \cite{Beterminlocal3d},
we were able to identify a range of densities such that cubic lattices
are  {\it locally} optimal at fixed density, but it is the first time -- according to
our knowledge -- that such phenomenon is observed at the level of
 the {\it global} minimizer.

\begin{figure}[H]
\begin{center}
	\includegraphics[width=8cm]{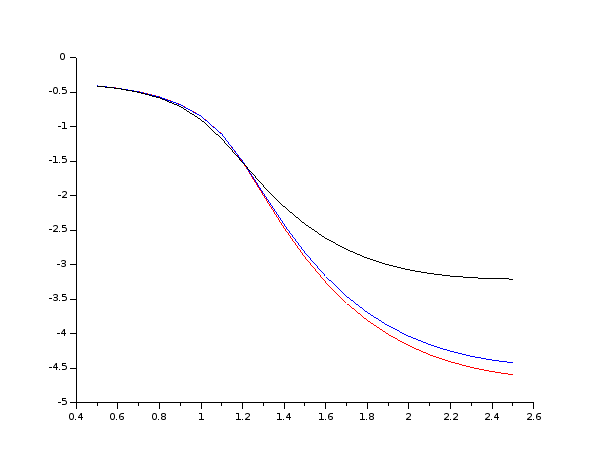}\quad \includegraphics[width=8cm]{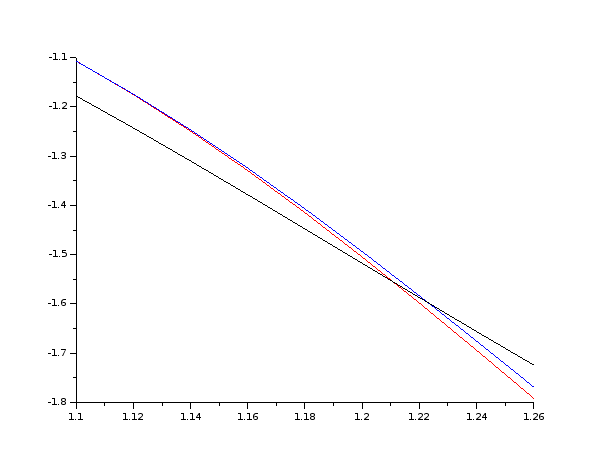}
	\end{center}
	\caption{Case \eqref{case1} in three dimensions. Plot of $\delta\mapsto \min_\lambda \mathcal{E}[\lambda L]$ for $L=\mathsf{D}_3$ (red), $L=\mathsf{D}_3^*$ (blue) and $L=\Z^3$ (black) on two different intervals.}
	\label{fig:3dLJTheta}
\end{figure}

\subsection{The power-law case $F(r)=r^t$}\label{sec:LJEAMPowerLaw}

  In this subsection,  we study the case where $F(r)=r^t$, $t>0$.
  Although 
  $F$ is not a one-well potential, 
   this case  turns  out to be  mathematically interesting. 
 Indeed, we are able to present  a special case where we can
explicitly compute $\min_\lambda \mathcal{E}[\lambda L]$ for any $L\in
\mathcal{L}_d(1)$. As we have seen above, this dimension reduction is
extremely helpful when one looks for the ground state of $\mathcal{E}$
in $\mathcal{L}_d$, especially for  $d=2$, since we can plot $L\mapsto
\min_\lambda \mathcal{E}[\lambda L]$ in the fundamental domain~$\mathcal{D}$.  

\subsubsection{A special power-law  case}

Let us now assume that
$$
F(r)=r^t, \quad  \ff  (r)=r^{-s}, \quad \phi(r)=ar^{-\alpha}- br^{-\beta},
$$
for $t>0$, $s>d$, $\alpha>\beta>d$, and $a,\,b>0$.  Therefore,  by \eqref{eq: zeta} we have, for any $\lambda>0$ and any $L\in \mathcal{L}_d(1)$, that
$$
\mathcal{E}[\lambda L]=\lambda^{-st}\zeta_L(s)^t + a\lambda^{-\alpha} \zeta_L(\alpha)-b\lambda^{-\beta} \zeta_L(\beta).
$$
For  a fixed lattice $L$, the critical points of $\lambda\mapsto \mathcal{E}[\lambda L]$ are the solutions of the following equation
\begin{align}\label{eq: CP}
b \beta \zeta_L(\beta) \lambda^{st+\alpha}-st \zeta_L(s)^t \lambda^{\alpha+\beta}-a\alpha \zeta_L(\alpha) \lambda^{st+\beta}=0.
\end{align}
 Solving this
equation in  closed form is impracticable out of  
a discrete set of parameter values. Correspondingly, 
comparing energy values  is even  more complicated than in
the pure  Lennard-Jones-type  case, which is already  challenging
when treated   in whole generality.

 Having pointed  out  this difficulty, we now focus on some additional
  specifications of the parameters, allowing to proceed further with
  the analysis. We have the following.

\begin{thm}[Special power-law case]
Let $\alpha,\beta, s$, and $t$ such that 
\begin{align}\label{eq: special-assu}	
 d<s,  \quad d<\beta < st < \alpha, \quad \textnormal{and}\quad \alpha=2st-\beta.
\end{align}
 Then, $\lambda^\mathcal{E}_{L}$ exists for all $L \in \mathcal{L}_d(1)$. Moreover,   $\lambda^\mathcal{E}_{L_d}L_d$  is a global minimizer in $\mathcal{L}_d$ of
$\mathcal{E}$,  now reading  
$$
\mathcal{E}[L]=\zeta_L(s)^t + a\zeta_L(\alpha)-b\zeta_L(\beta),
$$
 if and only if $L_d$ is a minimizer in $\mathcal{L}_d(1)$ of
\begin{align*}
e_*(L):&=-\frac{ \displaystyle C_1
         \frac{\zeta_L(s)^{2t}}{\zeta_L(\beta)}+C_2\zeta_L(s)^t\sqrt{c_1
         \frac{\zeta_L(s)^{2t}}{\zeta_L(\beta)^2} + c_2
         \frac{\zeta_L(\alpha)}{\zeta_L(\beta)}}+ C_3
         \zeta_L(\alpha)}{\displaystyle\left(\sqrt{c_1}
         \frac{\zeta_L(s)^t}{\zeta_L(\beta)}+ \sqrt{c_1
         \frac{\zeta_L(s)^{2t}}{\zeta_L(\beta)^2} + c_2
         \frac{\zeta_L(\alpha)}{\zeta_L(\beta)}}
         \right)^{\frac{\alpha}{\alpha-st}}}, \label{estar}
\end{align*}
where $C_i, c_j$, $i\in \{1,2,3\}$, $j\in \{1,2\}$, are positive constants defined by
\begin{equation} \label{constantCicj}
C_1:=\frac{st}{2b \beta}\left(\frac{st}{\beta}-1 \right),\quad C_2:=\frac{st}{\beta}-1 ,\quad C_3:=a\left(\frac{\alpha}{\beta}-1\right), \quad c_1:=\frac{s^2 t^2}{4 b^2 \beta^2},\quad c_2:=\frac{a\alpha}{b\beta}.
\end{equation}
\end{thm}
\begin{proof}
For any $L\in \mathcal{L}_d(1)$, any critical point of $\lambda\mapsto \mathcal{E}[\lambda L]$ satisfies (see \eqref{eq: CP}) 
$$
\lambda^{st+\beta}\left(b\beta \zeta_L(\beta)\lambda^{\alpha-\beta} - st\zeta_L(s)^t \lambda^{\alpha-st}-a\alpha \zeta_L(\alpha)  \right)=0.
$$
Since $\lambda>0$,  by  writing $X=\lambda^{\alpha-st}$  and using \eqref{eq: special-assu} we want to solve
$$
b\beta \zeta_L(\beta)X^2 - st\zeta_L(s)^t X-a\alpha \zeta_L(\alpha)=0,\quad X>0,
$$
for which the unique solution is
$$
X=\frac{st \zeta_L(s)^t + \sqrt{s^2 t^2 \zeta_L(s)^{2t} + 4a b \alpha \beta \zeta_L(\alpha)\zeta_L(\beta)}}{2 b \beta \zeta_L(\beta)}.
$$
Since $\alpha-st >0$ and  $b\beta \zeta_L(\beta)>0$, we find that the critical point is a minimizer and thus coincides with $\lambda^{\mathcal{E}}_L$ defined in \eqref{eq:optimal lambda}. More precisely, we have 
$$
 \lambda^{\mathcal{E}}_L=  \left( \frac{st \zeta_L(s)^t + \sqrt{s^2 t^2 \zeta_L(s)^{2t} + 4a b \alpha \beta \zeta_L(\alpha)\zeta_L(\beta)}}{2 b \beta \zeta_L(\beta)} \right)^{\frac{1}{\alpha-st}}.
$$
We hence get, for any $L\in \mathcal{L}_d(1)$, that
\begin{align*}
&\min_\lambda \mathcal{E}[\lambda L]=\mathcal{E}[\lambda^{\mathcal{E}}_L L]\\
&=(\lambda^{\mathcal{E}}_L)^{-st}\zeta_L(s)^t + a(\lambda^{\mathcal{E}}_L)^{-\alpha} \zeta_L(\alpha)-b(\lambda^{\mathcal{E}}_L)^{-\beta} \zeta_L(\beta)\\
&=(\lambda^{\mathcal{E}}_L)^{-\alpha}\left\{ \zeta_L(s)^t (\lambda^{\mathcal{E}}_L)^{\alpha-st} - b\zeta_L(\beta)(\lambda^{\mathcal{E}}_L)^{\alpha-\beta}+a\zeta_L(\alpha) \right\}\\
&= (\lambda^{\mathcal{E}}_L)^{-\alpha}\left\{ \zeta_L(s)^t (\lambda^{\mathcal{E}}_L)^{\alpha-st} - \frac{st\zeta_L(s)^t (\lambda^{\mathcal{E}}_L)^{\alpha-st}+a\alpha \zeta_L(\alpha)}{\beta}+a\zeta_L(\alpha) \right\}\\
&=(\lambda^{\mathcal{E}}_L)^{-\alpha}\left\{ \zeta_L(s)^t\left( 1-\frac{st}{\beta} \right)(\lambda^{\mathcal{E}}_L)^{\alpha-st} + a\zeta_L(\alpha)\left( 1- \frac{\alpha}{\beta}\right) \right\}\\
&=(\lambda^{\mathcal{E}}_L)^{-\alpha}\left\{ \zeta_L(s)^t\left( 1-\frac{st}{\beta} \right)\left(\frac{st \zeta_L(s)^t + \sqrt{s^2 t^2 \zeta_L(s)^{2t} + 4a b \alpha \beta \zeta_L(\alpha)\zeta_L(\beta)}}{2 b \beta \zeta_L(\beta)}  \right) + a\zeta_L(\alpha)\left( 1- \frac{\alpha}{\beta}\right) \right\}\\
&=(\lambda^{\mathcal{E}}_L)^{-\alpha}\left\{\frac{st}{2 b\beta}\left( 1-\frac{st}{\beta} \right) \frac{\zeta_L(s)^{2t}}{\zeta_L(\beta)}+ \left( 1-\frac{st}{\beta} \right) \zeta_L(s)^t\sqrt{\frac{s^2 t^2 \zeta_L(s)^{2t}}{4 b^2 \beta^2 \zeta_L(\beta)^2} + \frac{a\alpha \zeta_L(\alpha)}{b \beta \zeta_L(\beta)}   } + a\left( 1- \frac{\alpha}{\beta}\right)\zeta_L(\alpha)\right\},
\end{align*}
where  in  the fourth line we have used  the fact
that $\lambda^{\mathcal{E}}_L$ is a critical point of $\lambda\mapsto
\mathcal{E}[\lambda L]$, i.e., $b\beta
\zeta_L(\beta)(\lambda^{\mathcal{E}}_L)^{\alpha-\beta} - st\zeta_L(s)^t
(\lambda^{\mathcal{E}}_L)^{\alpha-st}-a\alpha \zeta_L(\alpha) =0$.  Note that by
assumption we have 
$$
1-\frac{st}{ \beta}<0,\quad 1-\frac{\alpha}{\beta}<0.
$$
It follows that, defining the  positive  constants $C_i,c_j$, $i\in \{1,2,3\}$, $j\in \{1,2\}$, as in \eqref{constantCicj}, that 
\begin{align*}
\min_\lambda \mathcal{E}[\lambda L]&= -(\lambda^{\mathcal{E}}_L)^{-\alpha}\left\{ C_1 \frac{\zeta_L(s)^{2t}}{\zeta_L(\beta)}+C_2\zeta_L(s)^t\sqrt{c_1 \frac{\zeta_L(s)^{2t}}{\zeta_L(\beta)^2} + c_2 \frac{\zeta_L(\alpha)}{\zeta_L(\beta)}}+ C_3 \zeta_L(\alpha) \right\}\\
&=- \frac{ \displaystyle C_1 \frac{\zeta_L(s)^{2t}}{\zeta_L(\beta)}+C_2\zeta_L(s)^t\sqrt{c_1 \frac{\zeta_L(s)^{2t}}{\zeta_L(\beta)^2} + c_2 \frac{\zeta_L(\alpha)}{\zeta_L(\beta)}}+ C_3 \zeta_L(\alpha)}{\displaystyle\left(\sqrt{c_1} \frac{\zeta_L(s)^t}{\zeta_L(\beta)}+ \sqrt{c_1 \frac{\zeta_L(s)^{2t}}{\zeta_L(\beta)^2} + c_2 \frac{\zeta_L(\alpha)}{\zeta_L(\beta)}}  \right)^{\frac{\alpha}{\alpha-st}}},
\end{align*}
which completes the proof.
\end{proof}

\subsubsection{Numerical investigations of the special power-law
  case in 2d and 3d}

We let $t\in (0,9/d)$ vary and fix 
$$
a=b=1,\quad \alpha=12,\quad \beta=6,\quad s=9/t,
$$
so that 
\begin{equation}
F(r)=r^t,\quad \ff  (r)=r^{-{9}/{t}},\quad \phi(r)=\frac{1}{r^{12}}-\frac{1}{r^6}.\label{case}
\end{equation}
Note that \eqref{eq: special-assu}	 holds under these assumptions. In two dimensions,  by testing as {$t\in (0,4.5)$} increases, we  observe  numerically the following: 
\begin{itemize}
\item If $t\in (0, t_1)$, $t_1\approx 1.605$, then $\mathsf{A}_2$ minimizes $e_*$ (see Figures~\ref{fig:s1} and  \ref{fig:sother});
\item If $t\in (t_1,t_2)$, where $t_2\approx 1.633$, then $\Z^2$ is a local  minimizer of $e_*$ but 
there seems to be  no minimizer for $e_*$ (see Figure~\ref{fig:sotherbis});
\item if  $t\in (t_2,4.5)$,   there seems to be  no minimizer for $e_*$, and $\Z^2$ is a saddle point (see Figure~\ref{fig:sother2}).
\end{itemize}


\begin{figure}[H]
\begin{center}
	\includegraphics[width=8cm]{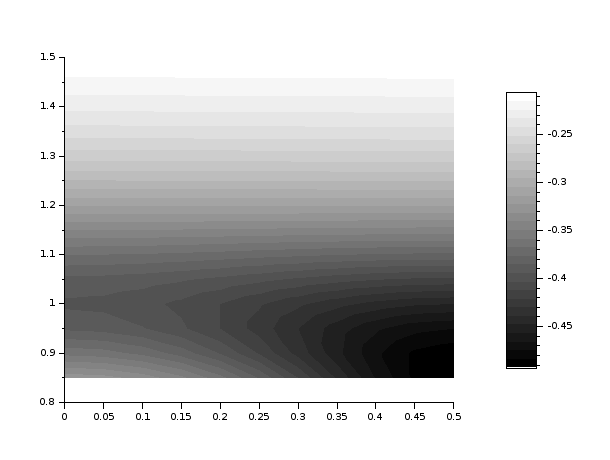}\quad 	\includegraphics[width=8cm]{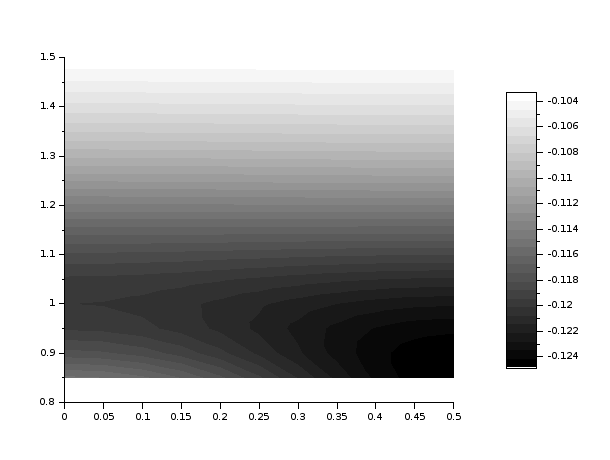}
	\end{center}
	\caption{Special power-law case \eqref{case} in two dimensions. Plot of
          $e_*$ on the fundamental domain $\mathcal D$.  For
          $t=1$ (left) and $t=1.5$ (right), the minimizer of $e_*$ is
          the triangular lattice $\mathsf{A}_2$ given by the point
          $(1/2,\sqrt{3}/2)$.}
	\label{fig:s1}
\end{figure}

\begin{figure}[h!]
\begin{center}
	\includegraphics[width=8cm]{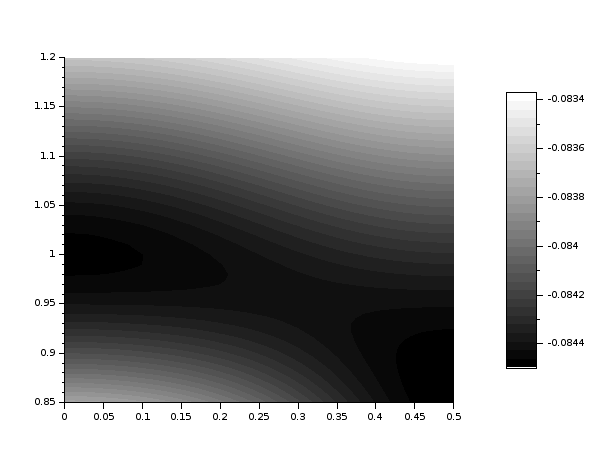}
	\end{center}
	\caption{Special power-law case  \eqref{case} in two dimensions. Plot of
          $e_*$ on the fundamental domain $\mathcal D$ for $t=1.605$. 
          The triangular lattice is the global minimizer of $e_*$
          whereas $\Z^2$ (given by the point $(0,1)$) is a local  minimizer. }
	\label{fig:sother}
\end{figure}

Similarly to the discussion of Subsection \ref{sec:LJEAMclassic},  for some choice of
parameters, a square lattice seems to be  locally  minimizing
  the EAM  energy,  at least within the range of
our numerical testing.  In \cite{Beterloc}, we have
identified a range of densities for which a square lattice is optimal
at fixed density. This seems however to be the first occurrence
of such minimality   among {\it all} possible lattices, without a
density constraint. Indeed, when minimizing among all lattices,
the square lattice  $\Z^2$   usually happens to be  a
saddle point, see, e.g., Figure \ref{fig:LJ126} for the Lennard-Jones
case.

\begin{figure}[h!]
\begin{center}
	 \includegraphics[width=7.9cm]{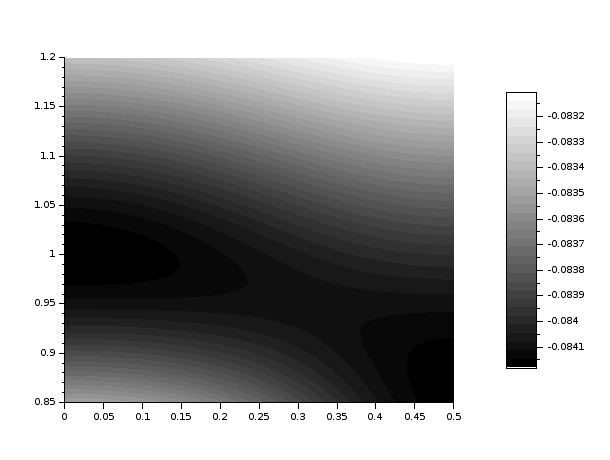} \quad
	 \includegraphics[width=7.9cm]{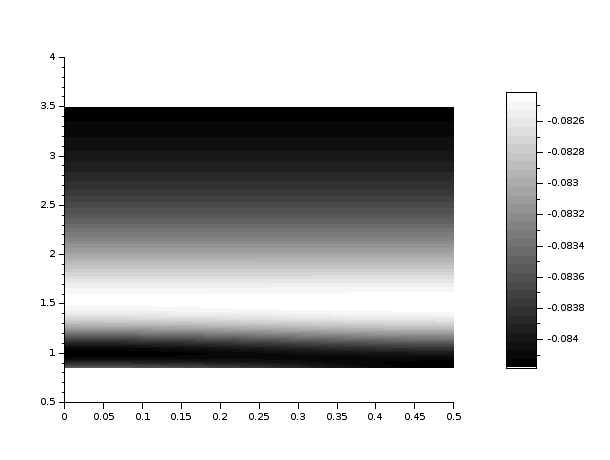}
	\end{center}
	\caption{Special power-law case  \eqref{case} in two dimensions. Plot of
          $e_*$ on the fundamental domain $\mathcal D$ for $t=1.606$.  The square lattice is a local  minimizer  of $e_*$ which
          does not have any  global  {minimizer}. Still,  $\mathsf{A}_2$ is a local minimizer. }
	\label{fig:sotherbis}
\end{figure}

\begin{figure}[H]
\begin{center}
	\includegraphics[width=8cm]{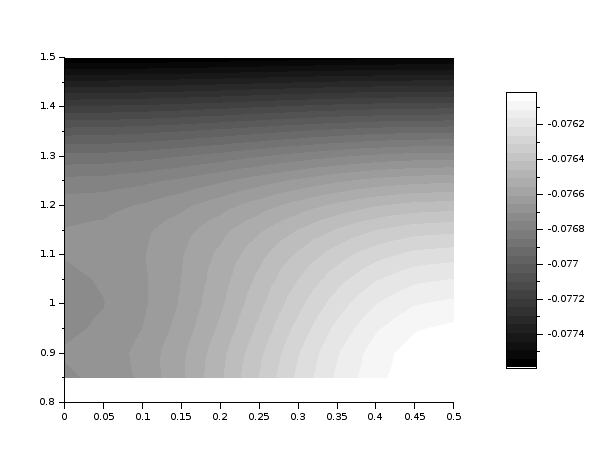}\quad
	 \includegraphics[width=8cm]{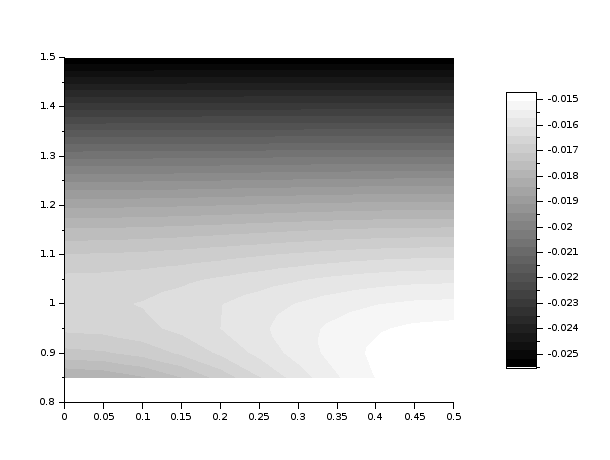}
	\end{center}
	\caption{Special power-law case  \eqref{case} in two dimensions.  Plot of
          $e_*$ on the fundamental domain $\mathcal D$.  For $t=1.632$ (left) the square lattice is a local  minimizer  of $e_*$ whereas $\mathsf{A}_2$ is a local  maximizer.  For $t=2$ (right) it seems that $e_*$ does not have any local minimizer and $\mathsf{A}_2$ stays a local  maximizer.  In both cases, there is no global minimum.}
	\label{fig:sother2}
\end{figure}

We have numerically investigated the three-dimensional case as well, 
comparing the energies of $L\in \{\Z^3, \mathsf{D}_3,
\mathsf{D}_3^*\}$. Figure~\ref{fig:s3d}  illustrates the numerical
results. We observe that there exist $t_1, t_2, t_3$, where $t_1\approx 1.5505$, $t_2\approx 1.5515$, and $t_3\approx 1.5647$ such that: 
\begin{itemize}
\item If $t\in (0,t_1)$, $e_*(\mathsf{D}_3)<e_*(\mathsf{D}_3^*)<e_*(\Z^3)$;
\item If $t\in (t_1,t_2)$, $e_*(\mathsf{D}_3)<e_*(\Z^3)<e_*(\mathsf{D}_3^*)$;
\item {If $t\in (t_2,t_3)$, $e_*(\Z^3)<e_*(\mathsf{D}_3)<e_*(\mathsf{D}_3^*)$};
\item {If $t\in (t_3,3)$, $e_*(\Z^3)<e_*(\mathsf{D}_3^*)<e_*(\mathsf{D}_3)$.}
\end{itemize}
When $t\to 0$, since $s=9/t\to \infty$ and $r^t\to 1$ for fixed $r>0$,
it is expected that the global minimizer of $\mathcal{E}$ in
$\mathcal{L}_3$ converges to the one of $E_\phi$, which  in turn  is expected to be a FCC lattice. This is supported by our numerics for $t<t_1$.
\begin{figure}[h!]
\begin{center}
	\includegraphics[width=8cm]{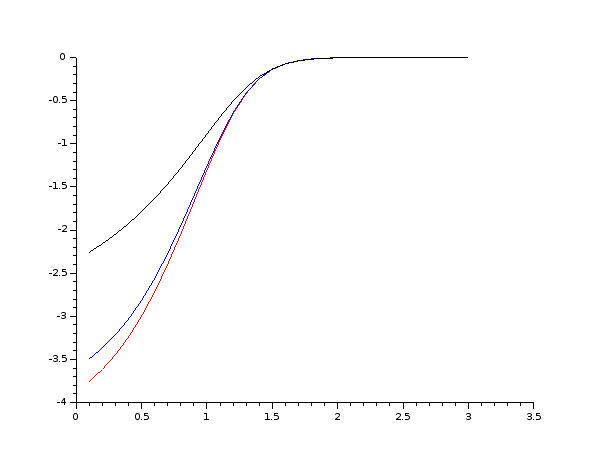}\quad
	 \includegraphics[width=8cm]{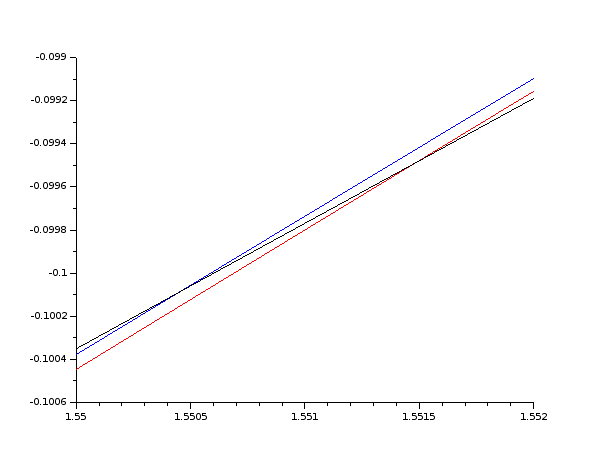}
	\end{center}
	\caption{Special power-law case  \eqref{case} in three dimensions.  Plot of
          $t\mapsto e_*(L)$ for $L=\mathsf{D}_3$ (red),
          $L=\mathsf{D}_3^*$ (blue) and $L=\Z^3$ (black) for $t\in
          (0,3)$. The graph on the right  is a close-up of the
            two transitions at $t_1$ and $t_2$.}
	\label{fig:s3d}
\end{figure}


\section*{Acknowledgments}
 MF  and US are   supported by  the  DFG-FWF 
 international joint  project FR 4083/3-1/I\,4354. MF is also
 supported by the Deutsche Forschungsgemeinschaft under Germany's Excellence Strategy EXC
  2044-390685587, Mathematics M\"unster:
  Dynamics--Geometry--Structure. 
 US and LB are supported by the FWF project F\,65. US is also supported
 by the FWF project  P\,32788.

\bibliographystyle{plain} %

\end{document}